\documentclass[prb,nofootinbib,twocolumn,superscriptaddress]{revtex4-2}
\usepackage{graphicx} 

\usepackage{float}
\usepackage[makeroom]{cancel}
\usepackage[utf8]{inputenc}
\usepackage{amsmath}
\usepackage[colorlinks,linkcolor=blue,citecolor=red]{hyperref}
\usepackage{color}
\usepackage{graphicx}
\usepackage{amsfonts}
\usepackage{amssymb}
\usepackage{mathtools}

\usepackage{theorem}
\usepackage{ulem}

\def\NoNumber#1{{\def\alglinenumber##1{}\State #1}\addtocounter{ALG@line}{-1}}

\newcommand{\cD}[0]{\mathcal{D}}

\newcommand{\cA}[0]{\mathcal{A}}

\newcommand{\T}[0]{\Delta\tau}
\newcommand{\Tx}[0]{\Delta\tau_x}
\newcommand{\Ty}[0]{\Delta\tau_y}

\newcommand{\FH}[0]{F^{H,\phi}}
\newcommand{\GH}[0]{G^{H,\phi}}
\newcommand{\GHt}[0]{\tilde{G}^{H,\phi}}

\newcommand{\FHS}[0]{F^{H_S,\phi_S}}
\newcommand{\GHS}[0]{G^{H_S,\phi_S}}
\newcommand{\GHtS}[0]{\tilde{G}^{H_S,\phi_S}}

\newcommand{\ket}[1]{\left\vert #1 \right\rangle}

\newcommand{\bC}[0]{\mathbb{C}}
\newcommand{\bR}[0]{\mathbb{R}}

\newcommand{\NP}[0]{\mbox{\bf NP}}
\newcommand{\Complete}[0]{\mbox{\bf complete}}
\newcommand{\BQP}[0]{\mbox{\bf BQP}}
\newcommand{\BPP}[0]{\mbox{\bf BPP}}
\newcommand{\QCMA}[0]{\mbox{\bf QCMA}}

\newcommand{\hard}[0]{\mbox{\bf hard}}
\newcommand{\hardness}[0]{\mbox{\bf hardness}}

\newcommand{\QMA}[0]{\mbox{\bf QMA}}
\newcommand{\MA}[0]{\mbox{\bf MA}}

\newtheorem{definition}{Definition}
\newtheorem{lemma}{Lemma}

\newtheorem{remark}{Remark}

\newtheorem{claim}[lemma]{Claim}
\newtheorem{fact}[lemma]{Fact}
\newtheorem{theorem}{Theorem}
\newtheorem{corollary}[lemma]{Corollary}
\newtheorem{conjecture}{Conjecture}

\newenvironment{proof}{{\bf Proof:}}{\hfill\rule{2mm}{2mm}}

\usepackage{newfloat,algcompatible}
\usepackage[size=small]{caption}
\usepackage{etoolbox}

\usepackage{newfloat}

\DeclareFloatingEnvironment[
    fileext=loa,
    listname=List of Algorithms,
    name=ALGORITHM,
    placement=tbhp,
]{algorithm}
\DeclareCaptionFormat{algorithms}{\vskip-15pt\hrulefill\par#1#2#3\vskip-6pt\hrulefill \vskip-6pt}
\makeatletter
\newcommand{\linethrough}{\mathpalette\@thickbar}
\newcommand{\@thickbar}[2]{{#1\mkern0mu\vbox{
    \sbox\z@{$#1#2\mkern-1.5mu$}%
    \dimen@=\dimexpr\ht\tw@-\ht\z@+2\p@\relax 
    \hrule\@height0.5\p@ 
    \vskip\dimen@
    \box\z@}}
}
\makeatother

\begin{document}

\title{Local Hamiltonian Problem with succinct ground state is  $\MA$-$\Complete$}

\author{Jiaqing Jiang}
\email{jiaqingjiang95@gmail.com}
\affiliation{Department of Computing and Mathematical Sciences, California Institute of Technology, Pasadena, CA, USA}

\begin{abstract}
Finding the ground energy of a quantum system is a fundamental problem in condensed matter physics and quantum chemistry. Existing classical algorithms for tackling this problem often
assume that the ground state has a succinct classical description,  that is,  a poly-size classical circuit for computing the amplitude.  Notable examples of succinct states encompass matrix product states,  contractible projected entangled pair states, and states that can be represented by classical neural networks. 
 We study the complexity of the local Hamiltonian problem with succinct ground state. 
We prove this problem is $\MA$-$\Complete$. 
 The Hamiltonian we consider is general and might not be stoquastic.  The $\MA$ verification protocol is based on the fixed node quantum Monte Carlo method, particularly the variant of the continuous-time Markov chain introduced by Bravyi, Carleo, Gosset and Liu. 
 Based on our work, we also introduce a notion of strong guided states, and conjecture that the local Hamiltonian problem with strong guided state is $\MA$-$\Complete$,  which will be in contrast with the $\QCMA$-$\Complete$ result of the local Hamiltonian problem with standard guided states.

\end{abstract}

\maketitle

\section{Introduction}

A fundamental question in quantum chemistry and condensed matter physics is finding the ground energy of a many-body system.
A main obstacle to designing efficient classical algorithms for finding ground energy, is the need for exponentially many
parameters to completely characterize the ground state. In practice, computational experts 
 often make an additional assumption on the many-body system,  that is, the ground state can be well-approximated by a succinct classical description. Here ``succinct" refers to that
 there exists a poly-size classical circuit which computes the amplitude of the ground state on any computational basis.
For instance, the Density Matrix Renormalization Group method (DMRG)~\cite{white1992density,white1993density,schollwock2005density,landau2015polynomial}, extensively used in quantum chemistry,  operates under the assumption that the ground states can be represented by matrix-product states (MPS)~\cite{vidal2003efficient}. 
 The  Hartree-Fock method~\cite{fischer1987general}, on the other hand, assumes that the ground states can be represented as Fermionic Gaussian states. For two-dimensional and higher-dimensional local Hamiltonians,  researchers have devised successful heuristic algorithms by representing the ground states by contractible projected entangled pair states~\cite{verstraete2008matrix,corboz2016variational,vanderstraeten2016gradient}. Additionally, there have been endeavors to model ground states using neural networks~\cite{carleo2017solving, sharir2022neural,gao2017efficient,carrasquilla2019reconstructing}.

In this work, we study the complexity-theoretic implications of the succinct ground state assumption, which is used in the above classical algorithms.
Recall that the decision version of the ground energy finding problem, often referred to as the local Hamiltonian problem (LHP)~\cite{kempe2006complexity}, is that given an $n$-qubit $k$-local Hamiltonian $H=\sum_{i=1}^m H_i$, 
two parameters $a,b$ where $b-a\geq 1/poly(n)$,  
determine  whether the ground energy of $H$, denoted as $\lambda(H)$, is less than $a$ or greater than $b$. 
Here we introduce a variant of LHP, denoted as \textit{LHP with succinct ground state},  where when $\lambda(H)\leq a$, we further assume that there exists a normalized ground state $\ket{\psi}$ which is \textit{succinct},  that is,  there exists a poly-size classical circuit $C_{\psi}(\cdot)$ which computes the amplitude of an unnormalized version of $\ket{\psi}$ as  $C_\psi(x)=  c\cdot\langle x|\psi\rangle$ for some $c$.
It is well-established that the standard LHP is $\QMA$-$\Complete$~\cite{kempe2006complexity}.  The central question (denoted as Q1) we ask is, how the assumption of a succinct ground state affects the complexity of LHP, or more precisely, what complexity class  LHP with a succinct ground state is.

  It is essential to note that the definition of succinct states is more general than MPS, Fermionic Gaussian states, and similar representations. For instance, the concept of a normalized state $\ket{\phi}$ being succinct does not necessarily imply efficient sample access to  $\ket{\phi}$,  that is,  the ability to efficiently generate sample $x$ with probability $|\langle x|\phi\rangle|^2$. This kind of sample access is a crucial element in many dequantization algorithms~\cite{ten1995proof,bravyi2022rapidly}, thus those dequantization algorithms do not apply to our setting.  
 Besides, the Hamiltonian $H$ we consider is general and may not be stoquastic, where stoquastic Hamiltonian~\cite{bravyi2006complexity} is known to be sign-free and general Hamiltonians suffer from the sign problem~\cite{hangleiter2020easing}. Due to those difficulties,
   for a general succinct state $\ket{\phi}$, it is  highly non-trivial to classically certify the energy,  that is,  $\langle \phi|H|\phi\rangle$. 
However,
our main result,
   Theorem \ref{thm:intro}, demonstrates that we can still efficiently dequantize the quantum verifier in LHP. More connections between Theorem \ref{thm:intro}, dequantization algorithms and stoquastic Hamiltonians are put into Section \ref{sec:RW}.
\begin{theorem}[Main theorem]\label{thm:intro}
	LHP with succinct ground state is $\MA$-$\Complete$.
\end{theorem}

In addition to exploring the underlying complexity aspects of classical algorithms, 
another motivation for  LHP with succinct ground state is to gain insight into the boundaries of quantum advantage.  Specifically, there is a widespread belief that quantum computers may provide an exponential advantage for quantum chemistry and condensed matter physics~\cite{cao2019quantum,mcardle2020quantum,bauer2020quantum}. One of the most promising pieces of evidence is the phase estimation algorithm~\cite{tilly2022variational}, which suggests that if we can prepare a \textit{guided state},   that is,  a state which has $1/poly(n)$ overlap with the ground state, 
 then there is an efficient quantum algorithm to estimate the ground energy. The guided states are suggested to be obtained from existing classical algorithms like DMRG or  Hartree-Fock, which has a succinct classical description. An important and natural question (denoted as Q2) that arises from this context is, with such classical access to the guided states, can we dequantize the phase estimation algorithm?

Existing literature partially refutes Q2, while many questions remain under research. Specifically,  one can define 
  the LHP with guided state~\cite{gharibian2022dequantizing}  as a variant of
  LHP, where when $\lambda(H)\leq a$, we further assume that there exists a guided state. For the standard definition of the guided states\footnote{We refer the standard guided state to be a state which has $1/poly(n)$ overlap with the ground state, and can be prepared by polynomial-size quantum circuit~\cite{gharibian2022dequantizing,weggemans2023guidable}.}, existing works do indicate a potential quantum advantage: LHP with guided state
   is proved to be  $\BQP$-$\Complete$ by Gharibian and Le Gall~\cite{gharibian2022dequantizing}, when the guided state is given; and 
proved to be
$\QCMA$-$\Complete$ by Weggemans, Folkertsma and Cade~\cite{weggemans2023guidable}, when the guided state is viewed as a witness.
When relaxing the precision  (the value of $b-a$ in LHP) from inverse-poly to constant, one can efficiently dequantize their algorithms,  that is,  in $\BPP$~\cite{gharibian2022dequantizing} and in $\NP$~\cite{weggemans2023guidable} respectively. 
Conditioned on $\BPP\subsetneq \BQP$ and $\NP\subsetneq \QCMA$, the above complexity results suggest that the quantum advantage for LHP is achieving higher precision.

However, these complexity results do not provide a comprehensive answer to Q2, since they overlooked the origin of the guided states.  As numerically investigated in \cite{lee2023evaluating}, if the
 guided states are obtained from existing classical algorithms,  for Hamiltonians in practice it is possible that those guided states are much better than the standard guided states, which will enable not only an efficient quantum algorithm for the LHP, but also an efficient classical algorithm\footnote{\cite{lee2023evaluating} also argued that with even standard guided states, for chemistry Hamiltonian which is a special case of the general local (Fermionic) hamiltonian, it is possible that the classical heuristics can efficiently estimate the ground energy to the desired precision.  Philosophically, \cite{lee2023evaluating} argued that if a classical algorithm is good enough to get a good guided state, which is non-trivially close to a ground state, then the problem itself might be classically tractable.
}. Inspired by \cite{lee2023evaluating}, from a complexity view we raise the following question (denoted as Q3): Is there a definition of ``strong guided states",  such that the complexity of LHP with such guided state drops from quantum to classical?

Here quantum refers to  $\QCMA$-$\Complete$ which is the complexity of LHP with standard guided states~\cite{weggemans2023guidable},  classical refers to $\MA$-$\Complete$, and  LHP with strong guided state is defined as a variant of LHP, where when
 $\lambda(H)\leq a$, we further assume that there exists a  ground state which admits a strong guided state.
 
 Recall that previous results~\cite{gharibian2022dequantizing,weggemans2023guidable} suggest that the quantum advantage for LHP is achieving higher precision. 
 Our result (Theorem \ref{thm:intro})  shows that even for inverse-poly precision, if the guided state in the Yes instance is extremely strong, i.e. is the ground state, the complexity of LHP with such guided states does drop to $\MA$-$\Complete$. 
 This opens the possibility that with a proper definition of strong guided states, even for inverse-poly precision, the LHP with strong guided states is $\MA$-$\Complete$, which will give an affirmative answer to Q3. This $\MA$-$\Complete$ will imply that heuristic randomized classical algorithms might tackle the corresponding LHP when the strong guided state is given, thus giving a partial answer to Q2.
Based on our work, we 
 propose a definition of \textit{strong guided states} to be the succinct approximation of ground state, which has entry-wise $1/poly(n)$ overlap with the ground state:

\begin{definition}[Strong guided state]\label{def:strong}
	We say an $n$-qubit normalized quantum state $\ket{\phi}$ admits a strong guided state, if there is a normalized state $\ket{\eta}$ which is succinct, and satisfies that 	  
	\begin{align}
		\forall x, \langle \eta|x \rangle \langle x|\phi\rangle \geq 
	\langle \phi|x \rangle \langle x|\phi\rangle/poly(n).\label{eq:strong}
	\end{align}
\end{definition}

Note that $\langle x|\phi\rangle $ can be a complex number, $\langle \eta|x\rangle
	\langle x|\phi\rangle \geq \langle \phi|x\rangle \langle x|\phi\rangle/poly(n)$ implicitly implies that $\ket{\eta}$ has information about the ``sign" of the ground state.  
 We conjecture that

\begin{conjecture}\label{conj:intro}
	LHP with strong guided state is $\MA$-$\Complete$. 
\end{conjecture}

Theorem \ref{thm:intro} is a special case of Conjecture \ref{conj:intro}, since a succinct ground state admits a strong guided state as itself. We hope our techniques will be useful for proving Conjecture \ref{conj:intro}.  Conjecture \ref{conj:intro} is known to be true if the Hamiltonian in the LHP is stoquastic~\cite{bravyi2014monte}. 
 The definition of strong guided states here is unarguably strong, which might not be the real reason that makes classical algorithms work~\cite{lee2023evaluating}. Here we view the LHP with strong guided state as a starting point of quantitatively understanding the boundary of quantum advantage for LHP ---  
 With a standard guided state, there is an efficient quantum algorithm for the LHP. However, when the guided state is too strong, there might also be efficient classical algorithms. What is the precise definition of strong guided states that enable classical algorithms thus break the potential quantum advantage?

For clarity,  we give a further remark on the conceptual relationship among previous work,  Theorem \ref{thm:intro} and Conjecture \ref{conj:intro}. 
Previous work~\cite{gharibian2022dequantizing,weggemans2023guidable}  suggests that quantum advantage for LHP is achieving higher precision. There are two intuitive reasons for such quantum advantage. The direct reason is that the overlap between the guided state and the ground state is $1/poly(n)$ which is relatively small compared to one, and a quantum algorithm which can directly perform quantum operations on guided states   seems more powerful than a classical algorithm which can only query the amplitude one by one. 
A further reason is that general Hamiltonians can have sign problems which makes it hard for classical algorithms to utilize the guided states. Otherwise, if the Hamiltonian is sign-free (stoquastic), even for inverse-poly precision, previous work~\cite{bravyi2014monte} already showed that 
the corresponding LHP with strong  guided state could be dequantized, that is,  $\MA$-$\Complete$ instead of $\QCMA$-$\Complete$. Our result  (Theorem \ref{thm:intro}) shows that for general Hamiltonians,  the sign problem can be resolved with the help of  extremely strong guided states. Our result suggests that   the sign problem might also be resolved with the help of certain strong guided states which lie in between standard and extremely strong guided states.
 
 \subsection{Related works}\label{sec:RW}

 In this section, we compare our work (Theorem \ref{thm:intro}) with some related works.
 
 \vspace{0.1em}
 \noindent\textbf{Stoquastic Hamiltonians and quantum Monte Carlo method.} 
 Technically, 
  our work is closely related to the quantum Monte Carlo method for stoquastic Hamiltonian and their generalizations, especially the fixed node Monte Carlo method~\cite{ten1995proof,bravyi2022rapidly}. For a Hamiltonian $H$ with a ground state $\ket{\psi}$, the goal of the quantum Monte Carlo method mentioned here is to define a Markov chain which can efficiently sample from the ground state $\ket{\psi}$,  that is, 
 outputting $x$ with probability $|\langle x|\psi\rangle|^2$.  To ease notations, we denote those quantum Monte Carlo  methods as \textit{sampling algorithms}, and denote the runtime required for the Markov chain to be close to the ground state sampling as \textit{mixing time.}    
 
 To begin with,
Bravyi and Terhal~\cite{bravyi2010complexity} first proved that if the Hamiltonian is frustration-free, stoquastic, and has a $1/poly(n)$ spectral gap, one can design a sampling algorithm with $poly(n)$ mixing time. 
 For a general Hamiltonian $H$, one can transform $H$ to a stoquastic Hamiltonian with some heuristic information by 
  the fixed node Monte Carlo method~\cite{ten1995proof}. 
   Roughly speaking, given a general Hamiltonian $H$ and an arbitrary  known (un-normalized) state $\ket{\phi}$ which we have query access to,
  \cite{ten1995proof} constructs  a stoquastic Hamiltonian which is called the fixed node Hamiltonian $\FH$, such that  
  \begin{itemize}
      \item[(1)]
 The ground energy of $\FH$ is always an uppper bound of the ground energy of $H$,  that is,  $$\lambda(\FH)\geq \lambda(H).$$
 \item[(2)]  If $\ket{\phi}$ is the ground state of $H$, then $\ket{\phi}$ is the ground state of $\FH$ and $\lambda(\FH)=\lambda(H).$
  \end{itemize}
   The drawback of $\FH$ is that its norm can be exponentially large, which might influence the mixing time. Even assuming  $\ket{\phi}$ is the ground state of $H$, there is no  rigorous bound for the mixing time of 
 most sampling algorithms based on $\FH$~\cite{ten1995proof,blinder2018mathematical}.
 A breakthrough is made by Bravyi \textit{et al.} recently \cite{bravyi2022rapidly}. Instead of defining a discrete-time Markov chain like most sampling algorithms, from $\FH$ Bravyi \textit{et al.} defined a continuous-time Markov chain. Furthermore,  Bravyi \textit{et al.} proved that if $\ket{\phi}$ is the true ground state of $H$,   and the Hamiltonian has a $1/poly(n)$ spectral gap together with some other good conditions, then the continuous-time Markov Chain has $poly(n)$ mixing time, that is, can efficiently sample from the ground state.

Our work builds from the Markov Chain in \cite{bravyi2022rapidly} and the properties of the fixed node Monte Carlo method~\cite{ten1995proof,bravyi2022rapidly}. The main difference between our work and \cite{bravyi2022rapidly}  is that we work on different tasks. \cite{bravyi2022rapidly} aimed for sampling from the ground state, with query access to a \textit{trusted} ground state. While we aim for testing ground energy, with query access to an \textit{adversarially} claimed ``ground state". The fact that the claimed ``ground state" can be adversarial
 is the main difficulty of solving LHP with succinct ground state.
If the  witness  (the claimed ``ground state" $\ket{\phi}$)  is trusted to be the true ground state, then the LHP with succinct ground state can be trivially solved by computing the ground energy directly. That is,
	$\lambda(H)=\frac{\langle x|H|\phi\rangle}{\langle x|\phi\rangle}$ for $x$ with $\langle x|\phi\rangle\neq 0$, 
	which can be computed efficiently since $H$ is sparse and 
	 $\langle x|H|\phi\rangle =\sum_y \langle x|H|y\rangle \langle y|\phi\rangle$.
  
Another remark is that since we work on a different task from \cite{bravyi2022rapidly}, we do not need the Hamiltonian to have a spectral gap or other additional conditions. Thus even in the Yes instance, one cannot efficiently sample from the ground state. Instead we utilize the promise gap to distinguish the Yes and No instances. Roughly speaking, we  define a ``Markov chain" that is well-defined for the Yes instances but becomes ill-defined for the No instances. Our verification algorithm distinguishes the Yes and No instances by testing whether the ``Markov chain" is well-defined. 
More can be seen in the proof overview,  that is, the  Section \ref{sec:proofoverview}.

 \vspace{0.2em}
\noindent\textbf{Dequantization algorithms.}
Conceptually,
our work  is close to several dequantization algorithms. In particular,  
 Gharibian and Le Gall~\cite{gharibian2022dequantizing} proved 
 when the promise gap is $1/poly(n)$,  (promise gap is  the value of $b-a$ in LHP), and for the setting where
 the guided state is given, 
 that LHP with succinct \textit{guided} state is $\BQP$-$\Complete$. When the promise gap is a constant,  one can efficiently dequantize the quantum phase estimation algorithm.  Our main theorem --- The $\MA$-$\Complete$ in Theorem \ref{thm:intro} --- does not contradict  with their  result  since we are working with
 succinct \textit{ground} state rather than  with
succinct \textit{guided} states, and on $1/poly(n)$ promise gap rather than constant promise gap.  Our setting is more similar to \cite{liu2020stoqma}, where  Liu proved that stoquastic LHP with succinct ground state is $\MA$-$\Complete$. The main difference between our work and Liu's is that we do not  assume the Hamiltonian to be stoquastic. 

Besides,
 although aiming for different tasks, one may wonder  whether the dequantization techniques initiated by Tang~\cite{tang2019quantum,chia2022sampling,chia2020quantum} work for our setting. Tang \textit{et al.}'s  settings are very different from ours, since they require sample access to the data, while the concept of succinct  does not imply sample access. To illustrate this point, consider a  3SAT formula $S(x)$, where
$x=x_1x_2...x_n$ are the values of the variables, $S(x)\in\{0,1\}$ denotes whether $x$ is  a satisfying assignment. By definition, 
 the $n$-qubit quantum state $\ket{\phi}$ which is a uniform superposition of all satisfying assignments is succinct, since the 3SAT formula can compute the amplitude of an unnormalized version of $\ket{\phi}$. However, it is not evident that how to efficiently obtain uniform samples of all satisfying assignments from the 3SAT formula.

\noindent\textbf{Verification of matrix products.}
There is an extensive study of the complexity of matrix verification~\cite{freivalds1979fast,buhrman2004quantum}, that is, given query access to  three matrices $A,B,C\in\bR^{2^n\times 2^n}$, verifying whether $AB=C$. Here we use $2^n$ as the matrix size for convenience  of comparison with our setting. 

 In particular, \cite{freivalds1979fast} gave a classical algorithm with high probability in time proportional to $2^{2n}$, and \cite{buhrman2004quantum}  gave a quantum algorithm in time $O(2^{5n/3})$.
 Our setting is conceptually related to verifying matrix product where the matrices have different sizes,  that is,  related to the question where  $\ket{\phi}\in \bR^{2^n}, H\in\bR^{2^n\times 2^n}$, and  testing whether the multiplication $\langle \phi|H|\phi\rangle \leq a$ or $\geq b$. In Theorem \ref{thm:intro} we claim that we have a $\MA$-type algorithm which only  needs $poly(n)$-queries instead of $poly(2^n)$. Note that this is achieved because we are testing ground energy
  rather than  testing matrix multiplication for arbitrary $\ket{\phi}, H$. Besides, in our setting the norm of $H\in \bR^{2^n\times 2^n}$ is bounded by $poly(n)$. A more detailed generalization  where our algorithm works  can be seen in Appendix \ref{appendix:matrix_verificaion}. This generalization only involves linear algebra and thus might be  easier to understand for readers outside quantum information.

\subsection{Proof overview}\label{sec:proofoverview}
 
 To prove Theorem \ref{thm:intro}, one needs to prove that LHP with succinct ground state is $\MA$-$\hard$ and is inside $\MA$. 
 The $\MA$-$\hardness$ directly comes from the Section 4 in
 \cite{bravyi2006complexity}, which is originally designed to prove that stoquastic Hamiltonian is $\MA$-$\hard$.  We  explain  why it implies LHP with succinct ground state is $\MA$-$\hard$ in more detail in Appendix \ref{appendix:MAhard}.

Our main contribution is proving LHP with succinct ground state is inside $\MA$. For better illustration, we make some simplifications here. Note that if an $n$-qubit state $\ket{\phi}$ is succinct, one can efficiently compute the ratio of any two amplitudes,  that is,  $\frac{\langle x|\phi\rangle}{\langle y|\phi\rangle}$. 
  Since we only use the ratios, to ease notations, here we assume that we can efficiently compute the amplitude $\langle x|\phi\rangle,\forall x$, that is, we have \textit{query access} to $\ket{\phi}$.
  For simplicity, in this section, we also assume that 
  $a=0, b=1/poly(n)$, and assume that $H$ and its ground state are real-valued. We always use the notations $\ket{\psi}$ for the true ground state, and $\ket{\phi}$ for an arbitrary state.

   To explain our $\MA$ verification protocol, 
  we begin with a direct algorithm that fails, then illustrate the ideas to  
  overcome the difficulties. Note that when 
  $a=0, b=1/poly(n)$,  given $H$ and query access to a succinct state $\ket{\phi}$,
   the LHP with succinct ground state is equivalent to test
   \begin{itemize}
       \item  Whether  $H\ket{\phi}=0$ with  $\lambda(H)=0$,
       \item Or $\langle \phi|H|\phi\rangle\geq 1/poly(n)$.
   \end{itemize}
    Here we assume 
   $\ket{\phi}\neq 0$ which can be checked easily  by providing a $x$ where $\langle x|\phi\rangle\neq 0$. 
  
   One cannot directly compute  $H|\phi\rangle$,
    since both $H$ and $\ket{\phi}$ are of exponential size. However, note that since $H$ is sparse,  one can efficiently check every row of  $H|\phi\rangle$,  that is, 
   \begin{align}
   		\langle x|H|\phi\rangle =\sum_y \langle x|H|y\rangle \langle y|\phi\rangle.
   \end{align}
Utilizing this,
   one may immediately come up with the following verification algorithm:
 
\begin{algorithm}
\caption{A direct algorithm}	\label{algo:direct}
\begin{algorithmic}[1]
		 \FOR{$t=1$ to $poly(n)$}		
            \STATE Randomly sample $x\leftarrow \{0,1\}^n$
			\STATE If $\langle x|H|\phi\rangle \neq 0$, return Reject
			 \hfill $\rhd$ Check row $x$
			\ENDFOR
		\STATE Return Accept
\end{algorithmic}
\end{algorithm}

Unfortunately,  Algorithm \ref{algo:direct} does not guarantee soundness. It checks the number of zeros in $H\ket{\phi}$ rather than the corresponding energy $\langle \phi|H|\phi\rangle$. Here is a simple counter-example: Let $H=I$, let $\ket{\phi}=|0\rangle^{\otimes n}$. $H$ is a No instance since $\lambda(H)=1$, while Algorithm \ref{algo:direct} accepts with probability close to $1$. Although one can easily rule out this case by adding preprocessing, it is easy to construct more complex counter-examples which still make Algorithm \ref{algo:direct} fail.

 Our key idea to circumvent the above problem, is instead of uniformly sampling the rows to be checked, we use the quantum Monte Carlo method to give different weights to different rows. Specifically, we add  checks in the quantum Monte Carlo algorithm which is originally designed for  sampling from the ground state.
 We first explain how \cite{bravyi2010complexity}  implicitly\footnote{\cite{bravyi2010complexity} does not explain  their results in the way we describe here.  In their setting the succinct classical circuit for the ground state is unknown, their main focus is constructing a circuit that can compute the ratio of the amplitudes,  which  obscures the idea we described here.}  used this idea to prove that
 LHP  w.r.t.  frustration-free stoquastic Hamiltonian  is $\MA$-$\Complete$, then explain how we adapt this idea to general Hamiltonians. A remark is that the $\MA$ verification protocol below is only a simplified version of the final protocol. In the final protocol, we need more witnesses and more checks. In particular, for the witness, we not only need a classical circuit for computing the amplitude of the claimed ground state, but also need a claimed ground energy, and a ``good" computational basis, which is a warm start of the Markov chain. More details are in Section \ref{Sec:pre}.

  Let us begin with the $\MA$ verification protocol when $H$ is stoquastic, where stoquastic means $\langle x|H|y\rangle\leq 0$ for all $x\neq y$. 
When $\lambda(H)=0$,    one can show that $H$ has  a ground state $\ket{\psi}$ whose amplitudes are real and non-negative~\cite{bravyi2010complexity}. 
  Given query access to the  ground state $\ket{\psi}$, one can connect $H$ to a Markov Chain~\cite{bravyi2010complexity}, whose transition matrix $P$ is defined as 
 \begin{align}\label{eq:px}
 	P_{x\rightarrow y} =\langle y|P|x\rangle := \langle y| I-\beta H|x\rangle \frac{\langle y|\psi\rangle}{\langle x|\psi\rangle},
 \end{align}
 where $\beta\leq 1/\|H\|$ is  to make $\langle y|I-\beta H|x\rangle\geq 0$, thus to make $\langle y|P|x\rangle\geq 0$ since it represents a probability. 
 Furthermore, when $\lambda(H)=0$,  every column of $P$ sums to $1$,  thus $P$ is a stochastic matrix and  a \textit{legal transition matrix}.
  The key connection between the Hamiltonian $H$ and the Markov chain $P$, is that the stationary distribution of $P$ is the distribution of sampling the ground state,  that is,   $x\sim |\langle x|\psi\rangle|^2$. Thus one can   sample from the ground state by a random walk w.r.t. $P$, where the mixing time depends on the spectral gap of $P$. A modification~\cite{bravyi2010complexity} of this sampling algorithm can be used
  to decide whether $\lambda(H)=0$  or $\lambda(H)\geq 1/poly(n)$
   when $H$ is stoquastic. Specifically, consider performing the random walk for  time $t=poly(n)$  --- In the Yes instance, $P$ is a legal transition matrix thus the random walk is always well-defined. In the No instance, for some $x$, the probability $P_{x\rightarrow y }$ is negative or the  sum of column $x$ of $P$ is not $1$, thus the random walk is not well-defined. The algorithm rejects if the random walk meets that $x$. Furthermore, in the No instance, one can show that
 the accepting probability will be upper bounded by a value proportional to $(1-\beta \lambda(H))^t$,  which decays exponentially fast since  $\lambda(H)\geq 1/poly(n)$ and $\beta= 1/poly(n)$ in the No instance.

 When $H$ is not stoquastic,   $\langle y|I-\beta H|x\rangle$ can have both positive and   negative entries, so does the entries in the   ground state $\ket{\psi}$. Thus the transition matrix defined in Eq.~(\ref{eq:px}) is not well-defined even in the Yes instance. 
 One way to handle this problem,  that is,  making the transition matrix well-defined, is to use the fixed node Monte Carlo method~\cite{ten1995proof} introduced in the related work section, which constructs a stoquastic Hamiltonian $\FH$ from a general Hamiltonian and an arbitrary state $\ket{\phi}$, with the property that $\lambda(\FH)\geq \lambda(H)$. Although $\FH$ is stoquastic and thus can be connected to a Markov chain,
  the main drawback of using $\FH$ directly is that the norm of $\FH$ can be exponentially large. 
 To define a legal probability in Eq.~(\ref{eq:px}),  that is, $P_{x\rightarrow y}\geq 0$, one needs to choose 
  the scaling factor $\beta\leq 1/\|H\|$ to be exponentially small, 
   which might in turn  make the accepting probability in the No instance, which is upper bounded by   $\sim(1-\beta \lambda(H))^t$, to be close to $1$. In other words,
   the exponentially small scaling factor 
   $\beta$ in  Eq.~(\ref{eq:px})
    hides any differences between the $H$ in the  Yes and No instance, thus making it hard to efficiently distinguish the two cases.

Instead of using a discrete-time Markov chain (DTMC), we build our protocol from 
the continuous-time Markov Chain (CTMC) by Bravyi \textit{et al.} \cite{bravyi2022rapidly}, which is based on the fixed node Hamiltonian $\FH$.  In case that the readers are not familiar with CTMC, we briefly explain some key concepts here. Recall that a DTMC is described by a transition matrix $P$, whose entries denote the transition probability. A CTMC is described by a generator matrix $G$, whose entries denote the transition rate, that is, for small $t$ and $x\neq y$, the probability of jumping from $x$ to $y$ is approximately $\langle y|G|x\rangle \cdot t$. Let the initial distribution be $\nu$, the distribution after
 evolving DTMC w.r.t.  time $t$ and transition matrix $P$  is  $P^{t}\nu$; The distribution after
 evolving CTMC w.r.t.  time $t$ and generator $G$ is   $e^{Gt}\nu$.  
 
 Now we review the key property of the CTMC~\cite{bravyi2022rapidly}, and explain  our $\MA$ verification protocol for the non-stoquastic $H$. The generator  of the CTMC~\cite{bravyi2022rapidly} is a rescaled version of $\FH$, that is, the 
 $\GHt$  defined as
	\begin{align}		\langle y| \GHt|x\rangle := 
			 &\lambda(\FH)\delta_{y,x}-\langle y|\FH|x\rangle \frac{\langle y|\phi\rangle}{\langle x|\phi\rangle}.
	\end{align}
	Note that $\GHt$ may also have an exponentially large norm. That means, if one simulates the  evolution  w.r.t. $\GHt$ and  time $t=poly(n)$ with a \textit{fixed} step-size,  that is,  $e^{Gt}=\left(e^{G\delta}\right)^{t/\delta}$, one need to choose $\delta$ to be exponentially small thus the algorithm runs in time $t/\delta$ and is in-efficient. 
 The key result in \cite{bravyi2022rapidly} is that if $\ket{\phi}$ is the true ground state of $H$, then
 \begin{itemize}
     \item[(1)] $\GHt$ is a \textit{legal generator}.
     \item[(2)]Furthermore, the CTMC w.r.t.  time $t=poly(n)$ can be simulated efficiently by Gillespie's algorithm.
 \end{itemize}
   Here  Gillespie's algorithm is another standard CTMC simulation method with \textit{varying} step-size.  In particular, \cite{bravyi2022rapidly} proved that if the Hamiltonian has a spectral gap and some other properties, 
  Gillespie's algorithm   converges to the stationary distribution of the CTMC in $poly(n)$ time. A very rough interpretation is that by using the  Gillespie's algorithm, 
 simulating this CTMC does not require the exponentially small scaling factor $\beta$ anymore, which is what we are searching for.
 Inspired by the verification algorithm  for stoquastic Hamiltonian described previously, for general Hamiltonian,  our $\MA$-verification algorithm is adding consistency checks in the CTMC ---
 In the Yes instance, since $\GHt$ is a legal generator, the CTMC is always well-defined. In the No instance, some parts of the generator are not well-defined. We reject if the random walk meets such parts.

 The completeness and soundness of our protocols come from the following observations.  For the completeness, a concern one may have is, compared to \cite{bravyi2022rapidly}, our setting does not have a large spectral gap, thus even in the Yes instance the CTMC cannot be close to the ground state sampling. We observe that the rapid mixing property is irrelevant for deciding the LHP. Instead what we need is 
 \begin{itemize}
     \item[(1)]  $\GHt$ is a legal generator thus the random walk is well-defined.
     \item[(2)]  One can efficiently simulate the CTMC   w.r.t. time $t=poly(n)$. 
 \end{itemize}
 Note that whether a generator is legal is independent of whether the generator has a large spectral gap.
  We will show that both of the above conditions hold   without a spectral gap~\cite{bravyi2022rapidly}.
  The soundness of our protocol mainly comes from the property of the Fixed node Hamiltonians, that is,  the ground energy of $\FH$ is always an upper bound of the ground energy of $H$: $$\lambda(\FH)\geq \lambda(H),$$ for any adversary state $\ket{\phi}$~\cite{ten1995proof}. Thus in the No case, 
 the accepting probability  decreases exponentially fast since $$\lambda(\FH)\geq \lambda(H)\geq 1/poly(n).$$

\subsection{Conclusion and future work}

In this manuscript, we study the underlying complexity question of existing classical algorithms for LHP,  which assumes that the ground state has a succinct classical description. More specifically, we define the local Hamiltonian problem with succinct ground state, and prove this problem is $\MA$-$\Complete$. A remark is that similar to stoquastic,  succinct is a basis dependent property.
It might be interesting to study the computational complexity of finding the basis which makes the ground state succinct (if such basis exists), where similar questions for stoquastic have been studied in \cite{marvian2019computational,klassen2019two,ioannou2020termwise}. In addition, 
the $\MA$-$\Complete$ result is established by simulating a continuous-time Markov chain  using Gillespie's algorithm. An interesting avenue for future work would be to investigate whether this result could also be demonstrated without this continuous-time technique.

As we illustrate in the introduction,
one of the most interesting open questions, is to relax the assumption from classical access of the ground state, to classical access of certain guided states. That is, giving a proper definition of strong guided states where the LHP with strong guided states is $\MA$-$\Complete$. 
Based on our work, we give a candidate definition of strong guided states as Definition \ref{def:strong}, and give the corresponding  Conjecture \ref{conj:intro}. The intuitive reason that we believe Conjecture \ref{conj:intro} is true, is because  according to Eq.~(\ref{eq:strong}), the strong guided state contains two key pieces of information: (a) The strong guided state knows the sign of every amplitude, which makes the problem closer to stoquastic Hamiltonian. (b) The strong guided state roughly knows every amplitude with certain errors. Note that although in this manuscript we choose a basic Markov chain which is  sensitive to errors and  cannot handle (b),  there  exist more complicated and robust quantum Monte Carlo methods. In particular, the Projection Monte Carlo algorithm in \cite{bravyi2014monte} can handle similar errors as in (b), where they proved if the Hamiltonian is stoquastic, then LHP with strong guided state is $\MA$-$\Complete$. The obstacle which  prevents us from  extending  our  Theorem \ref{thm:intro} to Conjecture \ref{conj:intro} by using the Projection Monte Carlo method \cite{bravyi2014monte}, is   the step before using the Markov chain.   
That is, the fixed node Hamiltonian $\FH$ builds from a  strong guided state $\ket{\phi}$ might not keep the promise gap. One way to tackle Conjecture \ref{conj:intro} is to develop or identify
 variants of fixed node Hamiltonians which keep the promise gap when assisted with the strong guided states.

Our primary motivation for defining strong guided states is to dequantize quantum algorithms, based on the conjectures~\cite{lee2023evaluating} that the guided states from existing classical algorithms are better than the standard guided states.
However, on the other hand, it is also interesting to ask, with strong guided states, can we design better quantum algorithms? Here better means fewer gates, lower depth quantum circuits, or much easier to implement in near-term quantum devices.

\subsection{Structure of the manuscript}

The structure of this manuscript is as follows. In Section \ref{sec:def}  we give notations and definitions which are used throughout this manuscript. In Section \ref{sec:preliminary} we introduce the techniques used in our verification protocol.  In Section \ref{sec:red} we briefly explain the argument from \cite{bravyi2022rapidly} that  w.l.o.g.  we can assume that the Hamiltonian and states are real-valued. In Section \ref{sec:CTMC} we review the continuous-time  Markov chain and  Gillespie's algorithm. In Section \ref{sec:FH} and Section \ref{sec:CTMCFH} we list  the construction and key properties of the fixed node Monte Carlo method~\cite{ten1995proof},
 and the variant of using continuous-time Markov chain~\cite{bravyi2022rapidly}.
 
  Finally, in Section \ref{sec:MA} we give our $\MA$ verification protocol and prove Theorem \ref{thm:intro}. In particular,
  in Section \ref{sec:CT} we assume that we have access to continuous-time randomness,  that is, 
we assume  that we can sample $u$ from uniform distribution of $[0,1]$ in $poly(n)$ time. In Section \ref{sec:DT} we substitute this assumption by discrete randomness and prove that the error induced by the discretization is small.

\section{Notations and Definitions}\label{sec:def}

\textbf{Notations.} 
We use $\bR,\bC$ to represent the real field and complex field.
Let $S\subseteq \{0,1\}^n$, we use  $\ket{\phi}\in \bR^{|S|}$ to denote a vector,
  $H,M \in\bR^{|S|\times |S|}$ to denote matrices, and  $x,y\in S$ to denote bit strings.
  The entries of $\ket{\phi}$ are indexed by $x\in S$, denoted by $\langle x|\phi\rangle$, similarly for $H,M$.

 we use $\|\ket{\phi}\|$ to denote the vector L2 norm.
 We say   $\ket{\phi}$ is normalized if  $\|\ket{\phi}\|=1$. We say an un-normalized state $\ket{\phi}$ is regularized if $\langle x|\phi\rangle\neq 0$ for all $x\in S$. 
 We use $\|M\|$ to denote its  spectral norm,  that is,  $$\|M\|:=\max_{\ket{\phi}\in \bR^{|S|},\ket{\phi}\neq 0}\frac{\|M\ket{\phi}\|}{\|\ket{\phi}\|}.$$ 
 We use $amax(M)$ to denote $\max_{x,y\in S} |\langle x|M|y\rangle|$. For a Hermitian $H$, we use $\lambda(H)$ to denote its ground energy,  that is,  the minimum eigenvalue. We use 
$\lambda_{max}(H)$ to denote its maximum eigenvalue. For $d=poly(n)$,
we say $H$  is $d$-sparse if
\begin{itemize}
    \item[(1)] Each row and each column of $H$ only have $d$ non-zero entries.
    \item[(2)] Given a row index $x\in S$ (or column index $y\in S$), there is a poly(n)-time algorithm which can list all the non-zero entries in  row $x$ (or column $y$),  that is,  list all $z\in S$ such that $\langle x|H|z\rangle\neq 0$, and similarly for column $y$.
\end{itemize}
  We say a Hermitian matrix $H$ or a vector $\ket{\phi}$ is real-valued if $\langle x|H|y\rangle \in \bR, \forall x,y\in S$, or $\langle x|\phi\rangle\in \bR,\forall x\in S$.  

A matrix $M$ is stochastic if 
\begin{itemize}
    \item[(1)] $\langle y|M|x\rangle\geq 0, \forall x$ and $y$. 
    \item[(2)] $\sum_{y} \langle y|M|x\rangle=1, \forall x$.
\end{itemize}
A Hermitian matrix $M$ is stoquastic if $$\langle x|M|y\rangle\leq 0, \forall x\neq y,x\in S,y\in S.$$

Given a vector $\ket{\phi}\in \bC^{|S|}$, we use $supp(\phi)$ to denote the positions of non-zero entries,  that is,  $$supp(\phi):=\{x\in S| \langle x|\phi \rangle\neq 0\}.$$  We use $\ket{\phi_{supp}}\in \bC^{|supp(\phi)|}$ to denote the vector obtained by deleting $0$ values in $\ket{\phi}$. For a vector $\ket{\phi}\in\bR^{|S|}$, we use $Diag(\phi)\in\bR^{|S|\times |S|}$ to denote the diagonal matrix where $$\langle x|Diag(\phi)|x\rangle:=\langle x|\phi\rangle,\text{ for }x\in S.$$ When $S=\{0,1\}^n$, we 
call $H,M\in\bC^{2^n\times 2^n}$  $n$-qubit operators, and $\ket{\psi}, \ket{\phi}\in\bC^{2^n}$  $n$-qubit vectors. 
When considering an $n$-qubit Hermitian $H$, we will always use $\ket{\psi} \in\bC^{2^n}$ to denote the ground state,  and $\ket{\phi}\in\bC^{2^n}$ to denote an arbitrary $n$-qubit state, which may or may not be the ground state. 
 Note that in this manuscript  $\ket{\psi}$  and $\ket{\phi}$  may be \textbf{unnormalized}.

An $n$-qubit $k$-local Hamiltonian is an $n$-qubit Hermitian operator $H\in\bC^{2^n\times 2^n}$, where $H=\sum_{j=1}^m H_j$, and $H_j$ acts nontrivially on at most $k$ qubits. We always assume that $k$ is a constant, $m=poly(n)$, and $\|H\|\leq poly(n).$

For any real number $\alpha$, we use $\lfloor \alpha\rfloor$  for the largest integer which is smaller than $\alpha$, and $\lceil \alpha\rceil$ for the smallest integer which is greater than $\alpha$. We use $\ln$ (or $\log$) to represent the logarithm w.r.t. the natural exponent $e$ (or $2$).

Let $w$ be a random variable. We use w.p. as an abbreviation for with probability.
We say $w$ is sampled from the uniform distribution $[0,1]$, if the probability density function is $p(w)=1,\forall w\in [0,1]$, and $p(w)=0$ for $w<0$ and $w>1$. We say $w$ is sampled from the exponential distribution of parameter $\lambda>0$, if the probability density function is $p(w)=\lambda e^{-\lambda w}$ for $w\geq 0$, and  $p(x)=0$ for $w<0$. 

We also consider the truncated discretized version of the exponential distribution:  
\begin{definition}
	Let $K$ be an integer, $\delta\in\bR$ be a small value,  and $\lambda>0$. We define the truncated discretized exponential distribution  $\cD_{K,\delta,\lambda}$ as a distribution over $\{k\delta\}_{k=0,1,..,K}$. In particular,
we say $w$ is sampled from $\cD_{K,\delta,\lambda}$
 if the probability is
\begin{align*}
	 &Pr_{w \sim \cD_{{K,\delta,\lambda}}}(	w=k\delta)=\\
  & \left\{
	 \begin{aligned}
	 	& \exp(-\lambda k\delta) - \exp(-\lambda (k+1)\delta), \text{ if $0\leq k\leq K-1$}\\
	 	& \exp(-\lambda k\delta), \text{ if $k=K$}\\
	 	&0 \text{ if  $k<0$ or $k>K$}
	 \end{aligned}
	 \right.
\end{align*}
\end{definition}
Note that 
$$\forall k\in\{0,...,K\}, Pr_{w\sim \cD_{{K,\delta,\lambda}}}(w\geq k\delta)=\exp(-\lambda k\delta).$$

\begin{definition}[Succinct quantum state]\label{def:ss}
	We say an $n$-qubit  normalized state $\ket{\phi}$ is succinct if there exists a poly$(n)$-size classical circuit $C_{\phi}$ which can compute the amplitude of  $\ket{\phi}$ up to a common factor,  that is,  $C_{\phi}(x)=c_{\phi}(n)\cdot \langle x|\phi\rangle$, where   $0<c_{\phi}(n)\leq 2^{poly(n)}$ is a function independent of $x$.
\end{definition}

Note that if $\ket{\phi}$ is succinct then one can compute the ratio of amplitudes of $\ket{\phi}$,  that is,  $\frac{\langle x|\phi\rangle}{\langle y|\phi\rangle}=\frac{C_{\phi}(x)}{C_{\phi}(y)}$. The requirement of $c_{\phi}(n)\leq 2^{poly(n)}$ comes from 
the fact that $C_{\phi}(x)$ should be efficiently described by $poly(n)$-bits.

\begin{definition}[LHP with succinct ground state]\label{def:Hss} 
 Given $(H,a,b)$ where $H=\sum_{i=1}^m H_i$ is an $n$-qubit $k$-local Hamiltonian, $\|H\|=poly(n)$, $m=poly(n)$, $k$ is a constant; $a,b$ are two parameters and $b-a\geq 1/poly(n)$. Besides, it is promised that one of the following holds:
	\begin{itemize}
		\item Yes instance: $\lambda(H)\leq a$, and there exists a  ground state $\ket{\psi}$ which is succinct.
		\item No instance: $\lambda(H)\geq b$.
	\end{itemize}
	The local Hamiltonian problem with succinct ground state is deciding whether $(H,a,b)$ is the Yes instance or the No instance.
\end{definition}

We implicitly assume that there is a sufficiently large polynomial $p(n)=poly(n)$, such that every value in Definition~\ref{def:Hss},  that is,  $\langle x|H|y\rangle,a,b,m, C_{\psi}(x)$ can be represented  by $p(n)$-bits. 
More clarification on precision can be found in Appendix \ref{appendix:precision} Remark \ref{remark:precision}. To ease analysis,  we also assume that we can use $poly(n)$-time to sample the discretized exponential distribution $\cD_{K,\delta,\lambda}$ exactly for parameters specified in Appendix \ref{appendix:precision} Remark \ref{remark:dist}.

Recall our main theorem is 
\begin{theorem}\label{thm:main}
Under the precision assumptions (Remark \ref{remark:precision} and Remark \ref{remark:dist} in Appendix \ref{appendix:precision}),
 LHP with succinct ground state is $\MA$-complete.
\end{theorem}
We will prove Theorem \ref{thm:main} in the following sections. A remark is that a slightly modified  proof can generalize  Theorem \ref{thm:main} from  local Hamiltonians  to  sparse Hamiltonians  whose spectral norm is bounded by $poly(n)$.

\section{Preliminaries}\label{sec:preliminary}

We first give two facts on the exponential function and the  spectral norm, which will be used repetitively. The proof can be found in Appendix \ref{appendix:fact}.
\begin{fact}\label{fact:ex}
	 For any $x\in[-1,1]$, $|e^{-x} - (1 -x)| \leq 2x^2.$
\end{fact}

\begin{fact} \label{fact:prin}
	Let $S'\subseteq S\subseteq \{0,1\}^n$ be two non-empty sets. Let $M\in\bR^{|S|\times |S|}$ be a Hermitian matrix.
	 Let $N\in\bR^{|S'|\times |S'|}$ be the submatrix of $M$ by restricting rows and columns in $S'$. Then  
	 \begin{align*}
	  \lambda(M)\leq \lambda(N)\leq \lambda_{max}(N) \leq \lambda_{max}(M).
	 \end{align*}
	 In particular, 
	 \begin{align*}
	 	\|N\|\leq \|M\|.	
	 \end{align*}

\end{fact}

\subsection{Reduction to real Hamiltonians}\label{sec:red}
In Definition \ref{def:Hss} the $n$-qubit $k$-local Hamiltonian $H$ and $n$-qubit state  $\ket{\psi}$ can have complex values. This section is to explain  w.l.o.g. we can assume that $H$ and $\ket{\psi}$ are real-valued by adding one ancilla qubit.  The following argument is simplified from  \cite{bravyi2022rapidly}.

Let $\ket{\phi}$ be any eigenstate of $H$ of eigenvalue $\alpha$.
Write $H=H_R+iH_I$ where $H_R,H_I$ are real-valued and $k$-local,
 write $\ket{\phi}=\ket{\phi_R}+i\ket{\phi_I}$ where $\ket{\phi_R},\ket{\phi_I}$ are real-valued. Define an $(n+1)$-qubit Hamiltonian $H'$ and $(n+1)$-qubit states $\ket{\phi^{'0}},\ket{\phi^{'1}}$ as
\begin{align*}
& H'=H_R\otimes I + H_I\otimes	 \begin{bmatrix}
	0&-1\\
	1&0
\end{bmatrix},\\
& \ket{\phi^{'0}} = \ket{\phi_R}\ket{0}+ \ket{\phi_I}\ket{1},\\
& \ket{\phi^{'1}} = -\ket{\phi_I}\ket{0}+\ket{\phi_R}\ket{1}.
\end{align*}
One can verify that $\ket{\phi^{'0}},\ket{\phi^{'1}}$ are orthogonal and are of eigenvalue $\alpha$ of $H'$. Besides, for any orthogonal $n$-qubit states
$\ket{\phi}$ and $\ket{\eta}$, the set $\{\ket{\phi^{'0}},\ket{\phi^{'1}}\}$ is orthogonal to  $\{\ket{\eta^{'0}},\ket{\eta^{'1}}\}$. 
Let $\ket{\phi}$ ranges over all $2^n$ eigenvectors of $H$, one will get a complete set of $2\cdot 2^n$ eigenvectors of $H'$. 
Thus $H$ and $H'$ have the same spectrum. In particular $H$ and $H'$ have the same ground energy. 

Given $H$, let $\ket{\psi}$  be its ground state, define the real-valued $H'$ and $\ket{\psi^{'0}}$ as above. Then
$H'$ is $(n+1)$-qubit, $(k+1)$-local. Further since $H$ is Hermitian, thus $H_R$ is symmetric and $H_I$ is anti-symmetric. Thus $H'$ is symmetric thus Hermitian.
 Note that one can efficiently calculate $\langle x'|H'|y'\rangle$, $\langle x'|\psi^{'0}\rangle$ with access to $\langle x |H|y\rangle,\langle x|\psi\rangle$, where  $x,y$ are the first $n$-bits of $x',y'$.  Thus  w.l.o.g. we assume that $H,\ket{\psi}$ are real-valued.

\subsection{Review of the CTMC and Gillespie's algorithm}
\label{sec:CTMC}

This section is a brief review of the continuous-time Markov chain and Gillespie's algorithm. More can be seen in Section 6 of \cite{andersonlecture} and other textbooks on Markov chain and random process~\cite{norris1998markov}. 
To ease notations,  in the following we abbreviate  the discrete state space, discrete-time Markov Chain as DTMC; and abbreviate the
discrete state space, continuous-time Markov Chain as CTMC. In this section, we use
 $S\subseteq\{0,1\}^n$ to denote the discrete  state space,   $P,G\in \bR^{|S|\times|S|}$ to denote matrices, and $x,y\in S$ to denote the states. For simplicity, we abbreviate $\sum_{x\in S}$ as $\sum_x$. 
 
 To introduce CTMC, we first review the concepts related to DTMC. 
 A DTMC with  discrete state space $S$ is a sequence of random variables $\{\xi(j)\}_{j=0,1...}$, where $\xi(j)\in S$ is the state at discrete time $j$. A DTMC is associated with   a  matrix $P\in \bR^{|S|\times|S|}$, which is called the transition matrix,  where $\langle y|P|x\rangle$ denotes the \textit{transition probability} from state $x$ to state $y$. $P$ is a \textit{legal} transition matrix  iff  $P$ is a stochastic matrix, that is, 
 \begin{align}
     \langle y|P|x\rangle\geq 0, \forall x,y;\text{ and  }\sum_y \langle y|P|x\rangle=1, \forall x.
 \label{eq:stochastic}
 \end{align}
  The DTMC w.r.t. $P$ is defined as the discrete-time stochastic process such that at any time $j$ where $\xi(j)=x$, one chooses the next state $\xi(j+1)$ to be $y$ w.p. $\langle y|P|x\rangle$.
 
 A  CTMC  with  discrete state space $S$ is a set of random variables $\{\xi(t)\}_{t}$, where $\xi(t)$ is the state at  time $t$, and $t$ can take values from a continuous interval, e.g. $t\in[0,1]$\footnote{The main difference between CTMC and DTMC is, in DTMC, the index of time is discrete, while in CTMC, the index of time is continuous.}. 
A CTMC is associated with  a matrix $G\in \bR^{|S|\times|S|}$, which is called the \textit{generator}, where $\langle y|G|x\rangle$ is called the \textit{transition rate} from state $x$ to state $y$. 
$G$ is a \textit{legal} generator iff  
$$\langle y|G|x\rangle\geq 0\text{ for }y\neq x;\text{ and  }\sum_y \langle y|G|x\rangle=0.$$  Note that when $G$ is a legal generator, one can verify that $\exp(Gt)$ is a stochastic matrix for any $t\geq 0$. A remark is that, to make the notations consistent with previous work~\cite{bravyi2022rapidly},  stochastic matrix in this manuscript (defined as in Eq.~(\ref{eq:stochastic})) has a column sum of 1 instead of a  row sum of 1.  
Also note that for a legal generator the diagonal elements are non-positive, since $$\langle x|G|x\rangle= -\sum_{y\neq x} \langle y|G|x\rangle\leq 0.$$  We say a distribution $\pi$ over the state space $S$ is a stationary distribution of the CTMC w.r.t. $G$ if 
$$\forall y,
	 \sum_x \langle y|G|x\rangle \pi(x)=0,$$
  which will imply 
	$$\sum_x \langle y|\exp(Gt)|x\rangle	\pi(x)=\pi(y), \forall y.$$ There are multiple ways to define a CTMC, here we define a CTMC via  Gillespie's algorithm as follows:
 Given a legal generator $G\in \bR^{|S|\times|S|}$, one can first define a DTMC named the embedded chain, where the transition matrix $P\in \bR^{|S|\times|S|}$ is defined as $\forall x,y\in S$,
 \begin{align}
 \langle y|P|x\rangle :=\left\{
 \begin{aligned}
 & 0, \text{ if $x=y$.}\\
 &  \frac{\langle y|G|x\rangle}{|\langle x|G|x\rangle|}, \text{ if $x\neq y$.}
  \end{aligned}
 \right.	
 \end{align}
 The CTMC w.r.t. $G$ can be defined via ``Embedded chain + Waiting time". Roughly speaking, 
 the CTMC w.r.t. $G$ is defined as the continuous-time stochastic process that when $\xi$  visits the state $x$ at time $\tau$, it will stay in the state $x$ for a ``waiting time" $\Delta \tau$, where $\Delta \tau$ is  sampled from the exponential distribution with parameter $|\langle x|G|x\rangle|$. Until time $\tau+\Delta \tau$,  it will move to the next state $y$ according to the embedded chain,  that is,  w.p. $\langle y|P|x\rangle$. More precisely, here we define
 the CTMC w.r.t. a legal generator $G$
 by the  
 Gillespie's algorithm, as shown in Algorithm \ref{algo:Gill}.

\begin{algorithm}
\caption{Gillespie's algorithm$(G, x_{in},t)$}\label{algo:Gill}
	\begin{algorithmic}[1]
		\STATE $x\gets x_{in}$,$\tau\gets 0$, $\xi(0)\gets x_{in}$
		\WHILE{$\tau<t$}
		\IF {$|\langle x|G|x\rangle|=0$}
			\STATE Set $\xi(s)=x$ for all $s\in (\tau,t]$
			\STATE $\tau \gets t$
		\ELSE
		\STATE Sample $u\in[0,1]$ from the uniform distribution  $[0,1]$
		\STATE $\Delta \tau\gets \frac{\ln(1/u)}{|\langle x|G|x\rangle|}$ \label{line:dt}
		\STATE Set $\xi(s)=x$ for all $s\in (\tau,\tau+\Delta\tau]$
		\STATE $\tau \gets \tau +\Delta\tau$
		\STATE Sample $y\in S \setminus\{x\}$ from the  probability distribution $\frac{\langle y|G|x\rangle}{|\langle x|G|x\rangle|}$
		\STATE $x\gets y$
		\ENDIF
		\ENDWHILE
		\STATE Return $x$
\end{algorithmic}
\end{algorithm}

 Lemma \ref{lem:legal} and Corollary \ref{cor:exp} are  properties of  Gillespie's algorithm.  Lemma \ref{lem:legal} describes the infinitesimal behavior of the CTMC. Roughly speaking, Lemma \ref{lem:legal} says that at any time $t$ where the current state is $x$, during a very short time $h$, the probability of staying in $x$ is approximately $1-|\langle x|G|x\rangle|h$. The probability of jumping to $y\neq x$ is approximately $\langle y|G|x\rangle h$.  What's more, the probability that $\xi$ changes its value more than twice during the  short time $h$ is negligible. The formal statement is as follows. 

\begin{lemma}[Gillespie's algorithm]\label{lem:legal} 
Let $S\subseteq \{0,1\}^n$ be the state space, $x_{in}\in S$, and
  $G\in\bR^{|S|\times |S|}$ be a legal generator. Let $h$ be an infinitesimal value.
Consider the random variable $\xi$ in Algorithm \ref{algo:Gill}  w.r.t. $(G,x_{in},t)$. 

For any $s<t$, 
 let $T(x,h)$ be the number of transitions  between time $[s,s+h]$
 conditioned on $\xi(s)=x$,  that is, the number of times that  $\xi$ changes its values. Then
\begin{align*}
	& Pr(T(x,h)=0) = 1 - |\langle x|G|x\rangle| h + O(h^2),\\
	& Pr(T(x,h)\geq 2) = O(h^2),\\
	& Pr(T(x,h)=1, \xi(s+h)=y, y\neq x) = \langle y|G|x\rangle h + O(h^2).
\end{align*}
\end{lemma}

Using Lemma \ref{lem:legal}, one can prove
\begin{corollary}[Gillespie's algorithm]\label{cor:exp}
Let $S\subseteq \{0,1\}^n$ be the state space, $x_{in}\in S$ and
 $G\in \bR^{|S|\times |S|}$ be a legal generator. 
 For any $s\in[0,t]$, the random variable $\xi(s)\in S$ generated by the Gillespie’s algorithm  w.r.t. $(G,x_{in},t)$ is distributed according to $$\pi_s(x):=\langle x|\exp(Gs)|x_{in}\rangle,\forall x.$$
\end{corollary}

In other words, Gillespie’s algorithm  w.r.t. $(G,x_{in},t)$ simulates a random process with transition matrix $\exp(Gs)$.
One can find the proofs for Lemma \ref{lem:legal} and Lemma \ref{cor:exp} in standard textbooks on Markov chain and random process.

\subsection{Fixed Node Hamiltonian}\label{sec:FH}

It is well-known that one can connect stoquastic Hamiltonian to Markov chain~\cite{bravyi2006complexity}. However,
in Definition \ref{def:Hss} the $n$-qubit $k$-local Hamiltonian $H$ can be general and might not be stoquastic.
The fixed node  
quantum Monte Carlo method~\cite{ten1995proof} is a method  that  transforms any real-valued Hamiltonian  to a stoquastic Hamiltonian. 

In this section, we give a brief review of the fixed node  
quantum Monte Carlo method, which will be used in our protocol later.
For technical reasons, in this section, we do not directly consider $n$-qubit Hamiltonians. Instead
 we consider 
 \begin{itemize}
     \item  A set $S\subseteq \{0,1\}^n$.
     \item A real-valued symmetric\footnote{Real-valued Hermitian matrix is real-valued symmetric matrix.} matrix $H\in \bR^{|S|\times |S|}$.
     \item A real-valued state $\ket{\phi}\in \bR^{|S|}$.
 \end{itemize}
The fixed node Hamiltonian 
$F^{H,\ket{\phi}}\in \bR^{|S|\times |S|}$ is defined as follows: For $x,y\in S$,
\begin{align}
\langle x |\FH|	y\rangle =
\left\{
\begin{aligned}
	& 0  &\text{if $(x,y)\in S^+$};\\
	& \langle x |H| y\rangle, &\text{if $(x,y)\in S^-$};\\
	& \langle x |H| x\rangle &+\sum_{z:(x,z)\in S^+} \langle x|H|z \rangle \frac{\langle z|\phi\rangle}{\langle x|\phi\rangle}\\
 &&\text{if } x=y.
\end{aligned}\label{eq:FH}
\right.
\end{align}
where
\begin{align*}
\begin{aligned}
	&S^+ =\{(x,y): x\neq y \text{ and } \langle \phi|x\rangle \langle x|H|y\rangle \langle y|\phi\rangle>0\},\\
	&S^- =\{(x,y): x\neq y \text{ and } \langle \phi|x\rangle \langle x|H|y\rangle \langle y|\phi\rangle\leq0\}.
\end{aligned}
\end{align*}

The key properties of $\FH$ are
\begin{lemma}[\cite{ten1995proof}, also in Lemma 2 of \cite{bravyi2022rapidly}]\label{lem:basic_FH} 
Given any real-valued symmetric matrix $H\in \bR^{|S|\times |S|}$ and any real-valued \textbf{regularized}\footnote{Otherwise Lemma \ref{lem:basic_FH} (1) is not a basis change thus not well-defined. In particular if $\ket{\phi}$=0,  then $\FH=H$ and $\FH$ might not be stoquastic.} unnormalized $n$-qubit state $\ket{\phi}\in \bR^{|S|}$\,\footnote{$\ket{\phi}$ can be arbitrary state. It is not necessary to be the ground state.}, define the fixed node Hamiltonian $\FH\in \bR^{|S|\times |S|}$ as above. Then
	\begin{enumerate}
		\item[(1)] $\FH$ is symmetric and real-valued thus Hermitian. Besides, $\FH$ is stoquastic modulo a change of basis $\ket{x}\rightarrow sign\left(\langle x|\phi\rangle\right)\ket{x}.$ 
		\item[(2)]  $\FH\ket{\phi}=H\ket{\phi}$.
		\item[(3)] $\lambda(\FH)\geq \lambda (H)$ for any $\ket{\phi}$. If further $\ket{\phi}$ is the ground state of $H$, then $\lambda(\FH)=\lambda(H)$ and $\ket{\phi}$ is the ground state of $\FH$,  that is, 
$$\FH\ket{\phi}=\lambda(\FH)\ket{\phi}.$$ 
		  \end{enumerate}
\end{lemma}
For completeness, we put a proof in Appendix \ref{appendix:FN}.

\subsection{CTMC related to $\FH$}\label{sec:CTMCFH}

 Define  $\FH\in \bR^{|S|\times |S|}$ as in Section \ref{sec:FH}. It is worth noting that the norm of $\FH$,  \textit{i.e.}  $\|\FH\|$, can be exponentially large, since $\langle x|\phi\rangle$ may be exponentially small for some $x$. As we discussed in the proof overview in Section \ref{sec:proofoverview}, the exponentially large norm influences the mixing time in the quantum Monte Carlo methods designed for stoquastic Hamiltonians. To handle this problem,
\cite{bravyi2022rapidly} describes a generic way of converting a
stoquastic Hamiltonian into a generator matrix of a CTMC. More details can be found in Lemma 3 of \cite{bravyi2022rapidly}. As a concrete application,  \cite{bravyi2022rapidly} applies this method to the Fixed node Hamiltonian and 
 defines a matrix $\GHt\in \bR^{|S|\times |S|}$ based on $\FH$. Specifically, for any $x,y\in S$, define
	\begin{align}\label{eq:GHt}
		\langle y| \GHt|x\rangle := 
			 &\lambda(\FH)\delta_{y,x}-\langle y|\FH|x\rangle \frac{\langle y|\phi\rangle}{\langle x|\phi\rangle},
	\end{align}
where $\delta_{y,x}=1$ iff $x=y$ and $\delta_{y,x}=0$ otherwise.
To avoid confusion, we emphasize here that we will define a slightly different $\GH$ in the following section in Eq.~(\ref{eq:GH}), which is used in the final verification protocol. For now, we temporarily focus on $\GHt$. It is worth noting that $\GHt$ is not symmetric. 
 The properties of $\GHt$ are summarized as follows.

\begin{corollary}[Corollary of Lemma 3 in \cite{bravyi2022rapidly}] \label{cor:basic_GH}
Given any real-valued symmetric matrix $H\in \bR^{|S|\times |S|}$,  and any real-valued  \textbf{regularized}  state $\ket{\phi}\in\bR^{|S|}$. Define
 $\GHt$  as in Eq.~(\ref{eq:GHt}).
	\begin{itemize}
		\item[(1)] The spectrum of $\GHt$ is the same as the spectrum of $\lambda(\FH)I-\FH$.
		\item[(2)] $\langle y| \GHt|x\rangle\geq 0$ for $y\neq x$. Further if $\ket{\phi}$ is an unnormalized ground state of $H$, then for any $x$, we have $\sum_y \langle y|\GHt|x\rangle =0$. 
		Thus $\GHt$ is a legal generator of a CTMC. 
		\item[(3)] If $\ket{\phi}$ is an unnormalized ground state of $H$.	Define $c=\|\ket{\phi}\|^2$ and $\pi(x)= |\langle x|\phi\rangle|^2/c$. Then for any $y$,
		$\sum_{x}\langle	y|\GHt|x\rangle \pi(x) =0$, thus   $\pi$ is a stationary distribution of the CTMC  w.r.t. $\GHt$.
		\end{itemize}
\end{corollary}

The key result in
\cite{bravyi2022rapidly} is proving Gillespie's
algorithm,  \textit{i.e.}  Algorithm \ref{algo:Gill}, can efficiently simulate the CTMC  w.r.t. $\GHt$. The formal statement is in Lemma \ref{lem:effS}.
For completeness, we put a proof in Appendix \ref{appendix:GH}.

\begin{lemma}[Lemma 4 in \cite{bravyi2022rapidly}]
\label{lem:effS}
	Given any real-valued $d$-sparse symmetric matrix $H\in \bR^{|S|\times |S|}$, and  real-valued \textbf{regularized}  unnormalized state $\ket{\phi}\in\bR^{|S|}$, where $\ket{\phi}$ is a ground state of $H$. Define
 $\GHt$  as in Eq.~(\ref{eq:GHt}).

	For any $x_{in}, t$, denote  $\kappa(x_{in},t)$ as  the number of transitions  in Algorithm \ref{algo:Gill} w.r.t. $(\GHt,x_{in},t)$,  that is,  the number of times that $\xi(s)$ changes its value. Then there exists $\hat{x}_{in}\in S$ such that $\langle \hat{x}_{in}|\phi\rangle\neq 0$ and the number of transitions satisfies that
	\begin{align}
		Pr \left( \kappa\left(\hat{x}_{in},t \right)\leq  d n^3  t\|H\|\right)\geq \frac{1}{2}.
	\end{align}
\end{lemma}

Note that  Lemma \ref{lem:effS} is highly non-trivial, since the norm of  $\GHt$ can still be exponentially large. In particular,
in Algorithm \ref{algo:Gill} line \ref{line:dt}, $\T$ may be exponentially small. Since the running time of Algorithm \ref{algo:Gill} is proportional to the number of transitions, exponentially small $\T$ may cause  
 Gillespie's algorithm w.r.t. ($\GHt$, $x_{in}$, $t$) to take exponentially long time, even for $t=O(poly(n))$. In the CTMC literature, it means that the CTMC may transit exponentially many times in a small time interval. However,
  Lemma \ref{lem:effS} proves that with high probability, the number of transitions is bounded, when $H$ is sparse and has a small norm.

\section{The $\MA$ verification protocol}\label{sec:MA}

In this section, we describe a $\MA$ verification protocol for deciding the local Hamiltonian problem with succinct ground state,  which is defined in Definition \ref{def:Hss}. 
Firstly in Section \ref{Sec:pre}, we will argue that w.l.o.g., we can assume that the instance $(H,a,b)$ from Definition \ref{def:Hss} satisfying 
$$a=0\text{ and }b\geq 1/poly(n).$$ 
Then in Section \ref{sec:CT} and Section \ref{sec:DT}, we will give two verification protocols based on different assumptions. In particular,
\begin{itemize}
	\item In Section \ref{sec:CT} we assume Assumption (i): We can sample $u$ from the uniform distribution $[0,1]$ in $poly(n)$ time. 
	\item In Section \ref{sec:DT} we substitute Assumption (i) by Assumption (ii): We assume that we can use $poly(n)$-time to sample from the truncated discretized exponential distribution  $\cD_{K,\delta,\lambda}$ with parameters specified in Remark \ref{remark:dist} in Appendix \ref{appendix:precision}. 
\end{itemize}
 In both Section \ref{sec:CT} and Section \ref{sec:DT}, we show that there is a $poly(n)$-time algorithm which takes inputs in the form of an $n$-qubit local Hamiltonian $H$, a $poly(n)$-size circuit $C_{\phi}$ for a succinct state $\ket{\phi}$, and an $n$-bit string $x_{in}$, such that 
 \begin{itemize}
 	\item  If $(H,a,b)$ is a Yes instance, there exists $C_{\phi}$ and
  $x_{in}$ s.t. the algorithm accepts w.p. $\geq constant$;
  \item  If $(H,a,b)$ is a No instance, then for any  $C_{\phi}$ and
  $x_{in}$, the accepting probability is exponentially small. 
 \end{itemize}

The proof in Section \ref{sec:CT} captures the key technical ideas. The proof in Section \ref{sec:DT} further addresses the errors introduced by discretization and  is our final proof for Theorem~\ref{thm:main}.

\subsection{Preprocessing}\label{Sec:pre}
Given an instance $(H^*,a^*,b^*)$, where $H^*=\sum_{j=1}^m H^*_j$  is an $n$-qubit $(k-1)$-local Hamiltonian, $m$ is the number of terms in $H$. Here $k$ is a constant, $m=poly(n)$ and $\|H^*\|=poly(n)$. 

By Section \ref{sec:red}, w.l.o.g. we assume that $H^*$ are real-valued symmetric $2^km$-sparse matrix, and is of size $2^n\times 2^n$. 
Besides, the ground state of $H^*$, denoted as  
   $\ket{\phi^*}\in \bR^{2^n}$, is also real-valued.

For the $\MA$ protocol, the witnesses are 
 \begin{itemize}
 	\item[(i)] A real value $\lambda^*$, which is supposed to be $\lambda(H^*)$\footnote{By Remark \ref{remark:precision} in Appendix \ref{appendix:precision}, $\lambda(H^*)$ can be represented by $poly(n)$ bits exactly.}. W.l.o.g. we assume that $\lambda^* \leq a^*$, otherwise we reject immediately.
 	\item[(ii)]  A $poly(n)$-size circuit $C_{\phi^*}$ which succinctly represents the state $\ket{\phi^*}$. $\ket{\phi^*}$ is supposed to be the ground state of $H^*$.
 	\item[(iii)] A computational basis $x^*_{in}\in\{0,1\}^n$,  which is supposed to satisfy Lemma \ref{lem:effS} w.r.t. $(\GHS,x_{in},t)$ which will be defined later. In particular,
  \begin{itemize}
      \item[$\bullet$]  $H_S,\ket{\phi_S},x_{in}$ are defined in the following paragraphs.
      \item[$\bullet$] $\GHS$ is defined in Eq.~(\ref{eq:GH}).
      \item[$\bullet$] Let
        \begin{align}
            &\epsilon:=b^*-a^*=1/poly(n),\\
            &t:=8 (n+p'(n))/\epsilon,
        \end{align}
     where $p'(n)=poly(n)$ is defined in Remark \ref{remark:G} in Appendix \ref{appendix:precision}.
  \end{itemize}
 \end{itemize}

 Specifically, define
\begin{align}
    &H:= H^*-\lambda^* I,\\
&\ket{\phi}:=\ket{\phi^*},\\
    &x_{in}:= x^*_{in}.
\end{align}
 The following properties hold:
\begin{itemize}
	\item  In the Yes instance, since $\lambda^*$, $\ket{\phi^*}$ are the ground energy and the ground state of $H^*$, we have $\lambda(H)= 0$ and $\ket{\phi}$ is the ground state of $H$. That is, $$H \ket{{\phi}}=0.$$
	\item In the No instance, since $\lambda(H^*)\geq b^*$ and $\lambda^*\leq a^*$. We have $$\lambda(H)\geq \epsilon.$$
\end{itemize}
Define 
$$S:=Supp(\ket{\phi}).$$ Define $H_{S}\in \bR^{|S|\times |S|}$ to be the submatrix of $H$ whose rows and columns are taken from the set $S$. 
 Similarly we define  $\ket{\phi_{S}}$ from $\ket{\phi}$. 
 The main reason we define $H_S,\ket{\phi_S}$ is to make the proof  rigorous. Feel free to assume that $\ket{\phi}$ is regularized, thus $H_S=H$, $\ket{\phi_S}=\ket{\phi}$.
 One can check that
 \begin{claim}\label{claim:HS}  Suppose $S\neq \emptyset$,
 	\begin{itemize}
 	 	\item $H_S$ is symmetric real-valued, $2^km$-sparse, where $m=poly(n)$ and $\|H_S\|=poly(n)$. Besides, $\ket{\phi_S}$ is real-valued and regularized.
	\item  In the Yes instance, $H_{S}\ket{\phi_S}=0$, $\lambda(H_S)=0$\footnote{Since $\lambda(H_S)\geq \lambda(H)=0$ by  Fact \ref{fact:prin}, and $0$ is achieved by $\ket{\phi_S}.$}.
	  $\ket{\phi_S}$ is a  ground state of  $H_S$.  
	\item In the No instance,  $$\lambda(H_S)\geq \lambda(H)\geq \epsilon =1/poly(n).$$
	\item For $x,y\in S$,  $\langle x|H_S|y\rangle$ and $C_{\phi_S}(x)$ can be easily computed in $poly(n)$ time, which is obtained by  directly computing $\langle x|H|y\rangle$ and $C_{\phi}(x)$.
	\end{itemize}
 \end{claim}
 \begin{proof}
 	It suffices to notice that by  Fact \ref{fact:prin},
 	when $S\neq \emptyset$,  we have 
  \begin{align*}
      &\lambda(H) \leq \lambda(H_S)\leq \lambda_{max}(H),\\
      &\|H_S\|\leq \|H\|\leq poly(n).
  \end{align*}
 \end{proof}

  Define $\FHS$ as in Eq.~(\ref{eq:FH}). Note that here 
instead of using $\GHtS$ as in Eq.~(\ref{eq:GHt}), we will use $\GHS$ defined below:
\begin{align}\label{eq:GH}
\langle y| \GHS|x\rangle :=
			 -\langle y|\FHS|x\rangle \frac{\langle y|\phi_S\rangle}{\langle x|\phi_S\rangle}. 
\end{align}

The main reason that we use $\GHS$ instead of $\GHtS$
is our setting is different from in \cite{bravyi2022rapidly}. More specifically, in \cite{bravyi2022rapidly}  $\ket{\phi_S}$ is trusted and is always the ground state, in this situation $$\lambda(\FHS)=\lambda(H_S).$$ 
By Lemma \ref{lem:basic_FH} (2), the value $\lambda(\FHS)$
can be computed  by 
\begin{align}\label{eq:lambda}
\lambda(\FHS)=\lambda(H_S)=\frac{\langle x|H_S|\phi_S\rangle}{\langle x|\phi_S\rangle}= \!\sum_y  \langle x|H_S|y\rangle \frac{\langle y|\phi_S\rangle} {\langle x|\phi_S\rangle},
\end{align}
where  $x\in S$ satisfying $\langle x|\phi_S\rangle \neq 0$. However, in our setting, in the No instance, $\ket{\phi_S}$ can be any adversary state. The ``$\lambda(\FHS)$" calculated by Eq.~(\ref{eq:lambda}) cannot be trusted and might break the soundness of the protocol. 

Here are some other remarks.
Note that in the Yes instance, Eq.~(\ref{eq:GHt}) and
Eq.~(\ref{eq:GH}) coincide since $\lambda(\FHS)=0$.   It is worth noting that  $\GHS$ is not symmetric.
 However, since $\FHS$ is symmetric,  by similar argument as Corollary \ref{cor:basic_GH} (1) we have
\begin{fact}\label{claim:GHH}
	 $\GHS$  is  always diagonalizable and has the same spectrum as   $-\FHS$.
\end{fact}

As explained in Remark \ref{remark:G} in Appendix \ref{appendix:precision}, there exists a polynomial $p'(n)\gg p(n)$ such that 
\begin{align*}
      &amax(\GHS)\leq 2^{p'(n)},\\
      &\|\GHS\|,\|\FHS\|\leq 2^{p'(n)+2n}.
  \end{align*}
Besides, if 
	$\langle x|\GHS|y\rangle\neq 0$, then 
 $$\langle x|\GHS|y\rangle\geq 1/2^{p'(n)}.$$

\subsection{Protocol with continuous-time randomness }\label{sec:CT}
In this section, we assume that
we can sample $u$ from the uniform distribution $[0,1]$ in $poly(n)$ time.
Let  $H_S$, $\ket{\phi_S}$  be the Hamiltonian and the state after the preprocessing procedure in Section \ref{Sec:pre} Claim \ref{claim:HS}.

  The verification protocol is shown in Algorithm \ref{algo:check}.  
The protocol is very similar to the truncated version of Gillespie's algorithm used in \cite{bravyi2022rapidly}. 
The key difference is, instead of returning $x$,  
we add a checking procedure in line \ref{aeq:GH}, and this algorithm returns Accept/Reject. Besides, we use $\GHS$ defined in Eq.~(\ref{eq:GH}) instead of $\tilde{G}^{H_S,\phi_S}$ in Eq.~(\ref{eq:GHt}).  

Compared to the setting in \cite{bravyi2022rapidly}, it is worth noting that
in our setting, $\ket{\phi_S}$ can be adversarial which might not be the ground state, thus $\GHS$ might not be a legal generator. The first thing to check is that Algorithm \ref{algo:check} is well-defined:
\begin{claim}
\label{claim:ct}
	Algorithm \ref{algo:check} is always well-defined and each line can be performed in poly(n) time.
\end{claim}
\begin{proof}	
 By line \ref{aeq:supp}, w.l.o.g. we assume that $S\neq \emptyset$ otherwise we reject immediately. 
 Then $\GHS,\FHS$ are well-defined since $S\neq \emptyset$.  Also note that line \ref{aeq:supp} implies a good $x_{in}$ which is not rejected by Line \ref{aeq:reject_x_0} should satisfy 
 $$\langle x_{in}|\phi\rangle\neq 0.$$ Furthermore, by definition of $\GHS$, that is, Eq~(\ref{eq:GH}), we have that all the $x$ visited by the algorithm in Line \ref{aeq:x_visit} satisfies $\langle x|\phi\rangle\neq 0$.
 
 For line  \ref{aeq:GH}, since $\GHS$ is $poly(n)$-sparse,  one can use poly-time to 
\begin{itemize}
    \item List all $z$ where $\langle z|\GHS|x\rangle\neq 0.$
    \item Compute $\sum_{z\in S} \langle z|\GHS|x\rangle$ and check whether one  of them is strictly negative. 
\end{itemize} 
  Similarly since $\GHS$ is $poly(n)$-sparse,
in line \ref{aeq:xy}, one can use poly-time to 
list all $y$ where $\langle y|\GHS|x\rangle\neq 0$, and efficiently sample $y$ from the distribution.

 Besides, when the conditions in line \ref{aeq:GH} are not satisfied, it is  guaranteed that
 the sampling procedure in line \ref{aeq:xy} is always well-defined, that is, 
 $$\sum_{y\neq x} \frac{\langle y|\GHS|x\rangle}{|\langle x|\GHS|x\rangle|}=1\text{ \!and\! }\frac{\langle y|\GHS|x\rangle}{|\langle x|\GHS|x\rangle|}\geq 0 \text{ for }y\neq x.$$ 
  Besides, by definition $\frac{\langle y|\GHS|x\rangle}{|\langle x|\GHS|x\rangle|}>0$ implies $y\in S$. Thus all the ``$x$" that $\xi$ visits in Algorithm \ref{algo:check} are all in $S$. Thus all the lines  querying entries  of $\GHS|x\rangle$ are well-defined.

\end{proof}

\begin{algorithm}[h!]
\caption{Checking$(H_S,\phi_S,x_{in})$}\label{algo:check}
	\begin{algorithmic}[1]
		\STATE $\kappa \gets 0$
        \hfill$\rhd$ {Record the number of transitions} 
		\STATE $x\gets x_{in}$, $\tau\gets 0$, $\xi(0)\gets x_{in}$ \label{aeq:initC}
		\STATE Set $t=  8(n+p'(n))/\epsilon$ and   $M= 2^k m n^3  t\|H\|$ \label{aeq:Mt}
		\IF {$x\not\in S$}\label{aeq:supp} \hfill$\rhd${A good $x_{in}$ should satisfy $\langle x_{in}|\phi\rangle\neq 0$}
		\STATE Return Reject \label{aeq:reject_x_0}
		\ENDIF
		\WHILE{$\tau<t$} \label{aeq:tauM}
		\IF{$C_{\phi_S}(x)$ is not represented by $p(n)$-bits}
			\STATE Return Reject 
			\hfill$\rhd${Check the format of $C_{\phi_S}(x)$}
		\ENDIF
		\IF{ $\sum_{z\in S} \langle z|\GHS|x\rangle \neq 0$ or $\exists y\neq x, y\in S$ s.t. $ \langle y|\GHS|x\rangle < 0$} \label{aeq:GH}
		\STATE Return Reject \label{aeq:99} \hfill$\rhd${Add a Check}
		\ENDIF		
		\IF { $\kappa\geq M$} \label{aeq:cm}
			\STATE Return Reject \label{aeq:cm1}
		\ENDIF \label{aeq:cm2}
		\IF {$|\langle x|\GHS|x\rangle|=0$}\label{aeq:xx}
			\STATE Set $\xi(s)=x$ for all $s\in (\tau,t]$ \label{aeq:ol1}
			\STATE $\tau \gets t$ \label{aeq:ol2}
		\ELSE
		\STATE Sample $u\in[0,1]$ from the uniform distribution $[0,1]$  \label{aeq:cr1}
        \NoNumber{$\Delta \tau\gets \frac{\ln(1/u)}{|\langle x|\GHS|x\rangle|}$} 
		\STATE Set $\xi(s)=x$ for all $s\in (\tau,\tau+\Delta\tau]$ \label{aeq:cr2}
		\STATE $\tau \gets \tau +\Delta\tau$ \label{aeq:cr3}
		\STATE Sample $y\in S\setminus\{x\}$ from the  probability distribution $\frac{\langle y|\GHS|x\rangle}{|\langle x|\GHS|x\rangle|}$ \label{aeq:xy}
		\STATE $x\gets y$ \label{aeq:x_visit}
		\STATE  $\kappa\gets \kappa+1$
		\ENDIF
		\ENDWHILE
		\STATE Return Accept \label{aeq:end}
\end{algorithmic}
\end{algorithm}

The performance of Algorithm \ref{algo:check} is summarized as follows. The proof is in Section \ref{sec:analysis}.
\begin{theorem}\label{thm:main_CT}
	For any $(H_S,\phi_S, x_{in})$, Algorithm \ref{algo:check} always runs in polynomial time. Besides,
 \begin{itemize}
     \item For the Yes instance, there exists $\phi_S$, $x_{in}$ such that Algorithm \ref{algo:check} accepts w.p. $\geq 1/2$.
     \item  In the No instance, $\forall$ $\phi_S$, $x_{in}$, Algorithm \ref{algo:check} rejects w.p. $\geq 1-2^{-n}$.
 \end{itemize}
\end{theorem}

\subsubsection{Analysis}\label{sec:analysis}

Note that the number of iterations in Algorithm \ref{algo:check}  is bounded by $M=poly(n)$. By Claim \ref{claim:ct} each line 
can be performed in $poly(n)$-time,  thus the algorithm always runs in polynomial time. 

In the following,  we say a string $x$ ``pass line \ref{aeq:GH}" if $x$ does \textbf{not} satisfy the conditions in line \ref{aeq:GH}, thus will not be rejected immediately in line \ref{aeq:99}. 
In the following we give the proof of completeness and soundness.

\vspace{1em}
\noindent\begin{proof}[\textbf{of Completeness of Theorem \ref{thm:main_CT}.}]
In the Yes instance, we have that $$\lambda(H_S)=0\text{ and }\ket{\phi_S}\text{ is the ground state.}$$
By Lemma \ref{lem:basic_FH} (3), $$\lambda(\FHS)=\lambda(H_S)=0.$$ Thus $\GHS$ coincides with $\GHtS$, see Eq.~(\ref{eq:GHt}) and
Eq.~(\ref{eq:GH}).  By Lemma \ref{lem:effS}, we know that there exists $x_{in}\in S$ such that   
	\begin{align}
 &\langle x_{in}|\phi_S\rangle\neq 0,\\
		&Pr \left(  \kappa\left(x_{in},t \right)\leq  M\right)\geq \frac{1}{2}.
	\end{align}
Besides, by Corollary \ref{cor:basic_GH} (2), we know that
\begin{align*}
    &\forall x\in S,\sum_{z\in S} \langle z|\GHS|x\rangle = 0,\\
    &\forall y\neq x, x\in S,y\in S, \langle y|\GHS|x\rangle \geq 0.
\end{align*}
Thus the checks in line \ref{aeq:supp}, \ref{aeq:GH} in Algorithm \ref{algo:check}
are always passed. The check in line \ref{aeq:cm} in Algorithm \ref{algo:check} is passed with probability greater than $\frac{1}{2}.$
Thus in summary, the probability of accepting is greater than $1/2$.

\end{proof}	

\noindent\begin{proof}[\textbf{of Soundness of Theorem \ref{thm:main_CT}.}]
W.l.o.g. assume line \ref{aeq:supp}  is  passed, otherwise the protocol rejects immediately. Thus we have 
$$\langle x_{in}|\phi_S\rangle\neq 0\text{ thus }\ket{\phi_S}\neq 0.$$
 In the No instance, we know 
 \begin{align}
 	\lambda(H_S)\geq \epsilon =1/poly(n),
 \end{align}
$\ket{\phi_S}$ can be an arbitrary adversary state, and $H_S\ket{\phi_S}\neq 0$.

In the No case $\GHS$ might not be a legal generator of a CTMC. To analyze the soundness, we  consider another algorithm without the check in line \ref{aeq:cm},  that is, deleting lines \ref{aeq:cm}, \ref{aeq:cm1} and \ref{aeq:cm2}.
 Denote this new algorithm as Algorithm \ref{algo:check}$^*$. 
 The accepting probability of Algorithm \ref{algo:check} can only be less than
 Algorithm \ref{algo:check}$^*$
 \footnote{Imagine running  Algorithm \ref{algo:check} without any checks,  that is, without line \ref{aeq:supp}, \ref{aeq:GH} and \ref{aeq:cm}, and we write down all the possible logs on a paper. Each log corresponds to a probability. What the checking procedure does is assigning reject to some logs. Less check, Less reject. It is possible that  Algorithm \ref{algo:check}$^*$ runs in exponential time. We do not care  about the efficiency of  Algorithm \ref{algo:check}$^*$, we only use it as a technique to bound the accepting probability of Algorithm \ref{algo:check}.}.

Define $S_{good}$ be the set of strings which pass line \ref{aeq:GH}.
 That is, 
 \begin{align}
 	 S_{good} :=&\{x\in S| \sum_{y\in S} \langle y|\GHS|x\rangle =0,\nonumber\\
   &\forall y\neq x, y\in S, \langle y|\GHS|x\rangle\geq 0\}.
 \end{align} 
 W.o.l.g. we assume that the first execution of line \ref{aeq:GH} is passed, otherwise the protocol rejects immediately. Thus we have
 $$
 \xi(0)=x_{in}\in S_{good}.
 $$

 Let $\GHS_{{good}}$ and $\FHS_{{good}}$ be the submatrix of $\GHS,\FHS$, where the row and column indices are in $S_{good}$. Let $\ket{\phi_{S_{good}}}$ be the state restricting $\ket{\phi_S}$ in $S_{good}$.
  It is worth noting that 
  \begin{align}
      \GHS_{{good}}\neq G^{H_{S_{good}},\phi_{S_{good}}},
  \end{align}
  since the diagonal elements are different. Instead, 
 \begin{align}
 		\GHS_{{good}}= Diag(\phi_{S_{good}})(-\FHS_{{good}})Diag(\phi_{S_{good}})^{-1}.\label{eq:72}
 \end{align}

 We claim that

\begin{claim}\label{cla:12}
\begin{align*}
    &\lambda_{max}(-\FHS_{good})\leq -\epsilon,\\
    &\|\exp(-\FHS_{good})\|\leq 1- \epsilon/2.
\end{align*}
where $\epsilon=1/poly(n)\leq 1/2$. 
\end{claim}
\begin{proof}
By definition $\ket{\phi_{S_{good}}}$  is regularized. 
	Note that by Lemma \ref{lem:basic_FH} (3), we have 
	\begin{align}
		\lambda(\FHS)\geq \lambda(H_S)\geq \epsilon.
	\end{align}
 Since  $-\FHS$ is symmetric thus Hermitian, by Fact \ref{fact:prin} we have 
 \begin{align*}
 	\lambda_{max}(-\FHS_{good})&\leq \lambda_{max}(-\FHS)\\
 	&= -\lambda(\FHS)\\
 	&\leq -\epsilon.
 \end{align*}
Note that  
$$\exp(-x)\leq 1-x/2\text{ for }x\leq 1/2,$$ and by definition of $\exp(\cdot)$, we know all eigenvalues\footnote{$\FHS_{good} $ is diagonalizable since it is Hermitian.} of $\exp(-\FHS_{good})$ is non-negative, thus 
	 we conclude that\footnote{$\exp(-\FHS_{good})$ is Hermitian, thus its  spectral norm is its maximum absolute value of eigenvalues. It is worth noting that  $\exp(\GHS_{good})$ is not Hermitian, although it has the same spectrum as $\exp(-\FHS_{good})$, they do not have the same  spectral norm.}  
  $$\|\exp(-\FHS_{good})\|\leq 1-\epsilon/2\text{ for }\epsilon\leq 1/2.$$
\end{proof}

In the following, we show that, although $\GHS$ does not correspond to a legal generator,  
Algorithm \ref{algo:check}$^*$ w.r.t. to $\GHS$ still have  similar infinitesimal properties as Lemma \ref{lem:legal}. The properties are summarized in Claim \ref{cla:pro}.  

To describe Claim \ref{cla:pro}, we define some notations. Consider running Algorithm \ref{algo:check}$^*$ w.r.t. $$(\GHS,\phi_S,x_{in}),$$ for some $x_{in}$\footnote{ 
A probability distribution of $x_{in}$ will not get a higher acceptance probability than one particular $x_{in}$ which maximize the accepting probability.}. Let $\tau_{end}$ be the value of $\tau$ when Algorithm \ref{algo:check}$^*$ returns Accept/Reject. Let $\xi(s), s\in[0,\tau_{end}]$ be the $\xi(s)$ in Algorithm \ref{algo:check}$^*$.  
We know that 
$$\xi(s)\in S_{good}, \forall s\in[0,\tau_{end}).$$  To clarify, $\tau_{end}$ and $\xi$ are random variables. For any fixed $s,t$, let $c_1$ be the first time that $\xi$ changes its value after time $s$, conditioned on $\tau_{end}\geq s$ and $\xi(s)=x$,  that is,  
	\begin{align*}
		c_1 =\min\{\eta:\eta>s,\xi(\eta)\neq x|\tau_{end}>s,\xi(s)=x\}.
	\end{align*}
Similarly, define $c_2$ to be  the second time that   $\xi$ changes its value after time $s$,  that is, 
	\begin{align*}
		c_2 =\min\{\eta: \eta>c_1, \xi(\eta)\neq \xi(c_1)|\tau_{end}>s,\xi(s)=x\}.
	\end{align*}

\begin{claim}\label{cla:pro}
	For any fixed $s< t$, any $x$, let $h$ be an infinitesimal  value. Use notations defined above. Conditioned on $\tau_{end}>s$ and $\xi(s)=x$, we have\footnote{We use $c_1$ to state this theorem instead of using the number of  transitions as in Lemma \ref{lem:legal}, since our algorithm ends immediately when it hits a not good string. It's a bit tricky to use the notion of  number of transitions in time $[s,s+h]$ here.} 
	\begin{itemize}
		\item The probability that $\tau_{end}\geq s+h$ and $\xi$ does not change value in time $[x,x+h]$ is
			\begin{align}
			&Pr(c_1\geq s+h | \tau_{end}>s,\xi(s)=x) \nonumber\\
   &= 1 - |\langle x|\GHS|x\rangle| h + O(h^2) \label{eq:xx}.
			\end{align}
		\item
		For $y\neq x, y\not\in S_{good}$, the probability that Algorithm \ref{algo:check} ends in time between $[s,s+h]$ by hitting $y$ is
		\begin{align}
			&Pr(c_1\leq s+h, \xi(c_1)=y|\tau_{end}>s,\xi(s)=x)\nonumber\\
   &= \langle y|\GHS|x\rangle h +O(h^2).
			\label{eq:xyb}
		\end{align} 
		\item For $y\neq x,y\in S_{good}$, the probability that $\xi$ hits $y$ in time $c_1\in [s,s+h]$, and keeps in $y$ in $[c_1,s+h]$ is
		\begin{align}
		&Pr(c_1\leq s+h, \xi(c_1)=y, c_2\geq s+h|\tau_{end}>s,\xi(s)=x)\nonumber\\
  &= \langle y|\GHS|x\rangle h +O(h^2)\label{eq:xyg}. 	
		\end{align}
	\item The probability of other events,  that is,  the probability that $\xi$ changes its value more than once\footnote{$\xi$ may or may not hit a not good string in $c_2$.} in $[s,s+h]$ is 
		\begin{align}
			Pr(c_2\leq s+h|\tau_{end}>s,\xi(s)=x)= O(h^2)\label{eq:el}.
		\end{align}
	\end{itemize}
\end{claim}
We rigorously prove Claim \ref{cla:pro} in Appendix \ref{appendix:df} using properties of the exponential distribution. On the other hand,
 one can intuitively imagine the correctness of Claim \ref{cla:pro} from Lemma \ref{lem:legal}: Although $\GHS$ may not be a legal generator, one can   consider another legal generator $G\in \bR^{|S|\times|S|}$ which is obtained by setting column $\langle *| \GHS |x\rangle$ to $0$ for $x\not\in S_{good}$.  Lemma \ref{lem:legal}  holds for  this legal generator $G$, and
 notice that the Algorithm \ref{algo:check}$^*$ behaves the same w.r.t. $\GHS$ as w.r.t. $G$ conditioned on it never hits $x \not\in S_{good}$.

With Claim \ref{cla:pro} we can prove:

\begin{claim}\label{claim:18}
	The accepting probability of Algorithm \ref{algo:check}$^*$ w.r.t. $(\GHS,\phi_S,x_{in})$
	is less than $1/2^n$.
\end{claim}

\begin{proof}
Consider the random process generated by running Algorithm \ref{algo:check}$^*$ w.r.t.  $(\GHS,\phi_S,x_{in})$. w.l.o.g.  assume that $x_{in}\in S$ otherwise line \ref{aeq:supp} rejects immediately. 

Let $P_x(s)$ be the probability that $\tau_{end}\geq s$ and $\xi(s)=x$.  Note that by definition of Algorithm \ref{algo:check}$^*$, 
$$\xi(\tau)\in S_{good}, \forall \tau\in [0,\tau_{end}].$$  Let $\ket{P_{good}(s)}\in \bR^{|S_{good}|}$ be the vector $$[...,P_x(s),...]^T\text{ for }x\in S_{good}.$$ Let $h$ be an infinitesimal  value, by Claim \ref{cla:pro},  we have for any $z\in S_{good}$,

\begin{align}
&P_z(s+h)
  =  \underbrace{P_z(s)\left(1-|\langle z|\GHS|z\rangle|h  
   + O(h^2)\right)}_\text{$\xi(s)=z$, stays in $z$ till $s+h$}\nonumber\\
   &+	\underbrace{\sum_{x:x\neq z, x\in S_{good}} P_x(s)\left(\langle z|\GHS|x\rangle h + O(h^2)\right)}_\text{$\xi(s)=x$, jump to $z$ between time $[s,s+h]$, stays in $z$ till $s+h$ } \nonumber \\
  & + \underbrace{O(h^2).}_{\substack{\xi(s)\in S_{good}, \text{ jump more than once in $[s,s+h]$} \\ \text{but finally stay $z$ in $s+h$}}} \label{eq:79}
\end{align}
Note that by the definition of $S_{good}$, we have 
\begin{align*}
\langle x|\GHS|x\rangle&= -\sum_{y\neq x,y\in S} \langle y|\GHS|x\rangle\\
&\leq 0.
\end{align*}
Thus we have for $z\in S_{good}$, $$-|\langle z|\GHS|z\rangle| =\langle z|\GHS|z\rangle.$$ Thus Eq.~(\ref{eq:79}) is equivalent to
\begin{align*}
	P_z(s+h) - P_z(s) &= \sum_{x\in S_{good}} \langle z|\GHS|x\rangle P_x(s) h +O(h^2)\\
	 & =\sum_{x\in S_{good}} \langle z|\GHS_{good}|x\rangle P_x(s) h +O(h^2),
	 \end{align*}
 where the last equality comes from the fact that $z,x\in S_{good}$. Thus

\begin{align*}
&P_z'(s) 
	 		= \langle z | \GHS_{good}|P_{good}(s)\rangle,\\
\Rightarrow &P'_{good}(s) = \GHS_{good} \ket{P_{good}(s)},	\\
 \Rightarrow &P_{good}(s) = \exp(\GHS_{good}s)\ket{x_{in}}.
\end{align*}
	Thus
	\begin{align}
	&Pr(\text{Algorithm \ref{algo:check}$^*$ Accept}) \label{eq:star_start}\\
 &= Pr(\tau_{end}>t)\nonumber\\
	& = \sum_{z\in S_{good}} P_z(t)\nonumber\\
	& =	\sum_{z\in S_{good}}  \langle z|\exp(\GHS_{good} t)|x_{in}\rangle\nonumber\\
	& =	\sum_{z\in S_{good}}  \langle z| Diag(\phi_{S_{good}})\exp(-\FHS_{good} t)Diag(\phi_{S_{good}})^{-1}|x_{in}\rangle\nonumber\\
	& =	\sum_{z\in S_{good}} \frac{ \langle z|\phi_{S} \rangle }{\langle x_{in}|\phi_S\rangle}\langle z|\exp(-\FHS_{good} t)|x_{in}\rangle.\nonumber
	\end{align}

By  Claim \ref{cla:12} and Remark \ref{remark:precision} in Appendix \ref{appendix:precision} we know that
\begin{align}
	&\sum_{z\in S_{good}} \frac{ \langle z|\phi_{S} \rangle }{\langle x_{in}|\phi_S\rangle}\langle z|\exp(-\FHS_{good} t)|x_{in}\rangle \label{eq:star_end}\\
 & \leq \sum_{z\in S_{good}} \left|\frac{ \langle z|\phi_{S} \rangle }{\langle x_{in}|\phi_S\rangle}\right|\cdot \left|\langle z|\exp(-\FHS_{good} t)|x_{in}\rangle \right|\nonumber\\
	&\leq 2^n \cdot 2^{3p(n)} \cdot (1-\epsilon/2)^t \label{eq:star_end_end} \\
	&\leq 2^{-n},\nonumber
\end{align}
for sufficiently large 
$$t\geq 8(n+p'(n))/\epsilon \geq 8(n+p(n))/\epsilon,$$ where $\epsilon\leq 1/2$. 
\end{proof}

Thus by Claim \ref{claim:18} finally we conclude that
 the accepting probability of Algorithm \ref{algo:check}, is smaller than the accepting probability of Algorithm \ref{algo:check}$^*$, which is smaller than $2^{-n}.$

\end{proof}

\subsection{Protocol with  discrete-time randomness}\label{sec:DT}

 Recall that we denote  $S\subseteq\{0,1\}^n$ as the state space,  $G\in \bR^{|S|\times|S|}$ as a matrix, and $x,y\in S$ as the states. Additionally, $p'(n)=poly(n)$ is a precision parameter explained in Appendix \ref{appendix:precision} Remark \ref{remark:G}, and
 $M$ is the upper bound of the number of transitions used in Algorithm \ref{algo:check},
 $$
  M = 2^kmn^3\cdot  8(n+p'(n))/\epsilon\cdot \|H\|  =poly(n).
 $$
 
 In this section, we replace the assumption from Section IV B—that one can sample uniformly from $[0,1]$  in $poly(n)$ time, and thus can sample from an exponential distribution  in $poly(n)$ time—with its discrete approximation. 
 In particular, we assume that we can use $poly(n)$-time to sample from the truncated discretized exponential distribution $\cD_{K,\delta,\lambda}$ as in Remark \ref{remark:dist} in Appendix \ref{appendix:precision}. 
 The value of $\lambda$ will be specified later  in the algorithm, while $K$ and $\delta$ are chosen to be
 \begin{align}   
    \delta &:= 2^{-2n}/M,\label{eq:Md1}\\
     K &:=  \lceil 2^{p'(n)+2n}  M  \left(n \ln 2 + \ln M\right)) \rceil.\label{eq:Md2}
 \end{align}

  Note that $\cD_{K,\delta,\lambda}$ serves as a discrete approximation of the exponential distribution in the following sense:
    \begin{claim}\label{claim:C2D}
         Sample a random variable $\Delta \tau$ according to the exponential distribution with parameter $\lambda$. Let $\Delta \tau_D$ be the rounded value of $\Delta \tau$, which is 
        the largest value in the set $\{k\delta\}_{k=0,...,K}$ that does not exceed $\Delta \tau$. Then the distribution of   $ \Delta \tau_D$ is $\cD_{K,\delta,\lambda}$.  
    \end{claim}
Here, the subscript $D$ in $\Delta \tau_D$ refers to the  ``discrete approximation".

The discretized $\MA$ verification protocol is derived by modifying Algorithm \ref{algo:check}: First, for clarity, we rename the variables  $(\tau,\Delta\tau,\xi)$ to $(\tau_D,\Delta\tau_D,\xi_D)$. 
 Due to the discretization error, in line \ref{aeq:Mt} we change the value of $t$ from $8(n+p'(n))/\epsilon$ to  $(- 2^{-n})+8(n+p'(n))/\epsilon.$   
Finally, we replace the continuous-time process (lines \ref{aeq:cr1}, \ref{aeq:cr2} and \ref{aeq:cr3} in Algorithm \ref{algo:check}) with its discrete approximation, as shown below: 
\begin{center}
\begin{algorithm}[H]
    \begin{algorithmic}
        \STATE Sample $\Delta \tau_D$ from $\cD_{K,\delta,|\langle x|\GHS|x\rangle|}$
		\STATE Set $\xi_D(s)=x$ for all $s\in (\tau_D,\tau_D+\Delta\tau_D]$ 
		\STATE $\tau_D \gets \tau_D +\Delta\tau_D$ 
    \end{algorithmic}
\end{algorithm}
\end{center}

We denote the discretized $\MA$ verification protocol  as Algorithm \ref{algo:check}D.
The performance of Algorithm \ref{algo:check}D is summarized as follows.

\begin{theorem}\label{thm:main_DT} 
	For any $(H_S,\phi_S, x_{in})$, Algorithm \ref{algo:check}D always runs in polynomial time. Besides,
 \begin{itemize}
     \item For the Yes instance, there exists $\phi_S$, $x_{in}$ such that Algorithm \ref{algo:check}D accepts w.p. $\geq 1/2-2^{-n}$. 
     \item  In the No instance, $\forall$ $\phi_S$, $x_{in}$, Algorithm \ref{algo:check}D rejects w.p. $\geq 1-3\cdot 2^{-n}$.
 \end{itemize}
\end{theorem}

\subsubsection{Connecting Algorithm \ref{algo:check} and Algortihm \ref{algo:check}D}

 Based on Claim \ref{claim:C2D} we can interpret the discretized distribution $\cD_{K,\delta,\lambda}$ as being derived from rounding the exponential distribution.  Using this perspective\footnote{We thank the anonymous reviewers for suggesting this connection, which significantly simplifies the proof.},  we will prove Theorem \ref{thm:main_DT} by showing that Algorithm \ref{algo:check}D is a good discrete approximation of  Algorithm \ref{algo:check}. 

\begin{algorithm}
\caption{Checking\_Compare$(H_S,\phi_S,x_{in})$}\label{algo:checkCD}
	\begin{algorithmic}[1]
		\STATE  $\kappa\gets 0$
        \hfill$\rhd${Record the number of transitions}
\STATE $x \gets x_{in}$,$\tau \leftarrow 0$,  $\xi (0)\gets x_{in}$ 
\NoNumber{$\tau_D\gets 0$, $\xi_D(0)\gets x_{in}$
} 
          \STATE \hfill$\rhd${Do not specify $M,t$.}
		\IF {$x\not\in S$}\label{aeq:suppD}
 \hfill$\rhd${A good $x_{in}$ should satisfy $\langle x_{in}|\phi\rangle\neq 0$}
		\STATE Return Reject
		\ENDIF
		\WHILE{True} 
		\IF{$C_{\phi_S}(x)$ is not represented by $p(n)$-bits}
			\STATE Return Reject 
			\hfill$\rhd${Check the format of $C_{\phi_S}(x)$.}
		\ENDIF
		\IF{ $\sum_{z\in S} \langle z|\GHS|x\rangle \neq 0$ or $\exists y\neq x, y\in S$ s.t. $ \langle y|\GHS|x\rangle < 0$} \label{aeq:GHD}
		\STATE Return Reject \hfill$\rhd${Add a Check}
		\ENDIF		
	\STATE \hfill$\rhd${Delete the check $\kappa\geq M$}	
			\STATE{} 
		\STATE
		\IF {$|\langle x|\GHS|x\rangle|=0$}\label{aeq:xxD}
			\STATE Set $\xi(s)=x$ for all $s\in (\tau,+\infty]$; 
                \NoNumber{Set $\xi_D(s)=x$ for all $s\in (\tau_D,+\infty]$;}
			\STATE $\tau \gets +\infty$
            \NoNumber{$\tau_D \gets +\infty$}
		\ELSE
		\STATE Sample $u\in[0,1]$ from  uniform distribution $[0,1]$ \label{aeq:cr1_CD}
\NoNumber{ $\Delta \tau\gets \frac{\ln(1/u)}{|\langle x|\GHS|x\rangle|}$}
\STATE Set $\xi(s)=x$ for all $s\in (\tau,\tau+\Delta \tau]$
\STATE $\tau\leftarrow \tau+\Delta \tau$
\NoNumber{Let $\Delta \tau_D$ be the rounded value of $\Delta \tau$, which is 
        the largest value in the set $\{k\delta\}_{k=0,...,K}$ that does not exceed $\Delta \tau$.}
\NoNumber{Set $\xi_D(s)=x$ for all $s\in (\tau_D,\tau_D+\Delta\tau_D]$} 
\NoNumber{$\tau_D \gets \tau_D +\Delta\tau_D$} 
		\STATE Sample $y\in S \setminus\{x\}$ from the probability distribution $\frac{\langle y|\GHS|x\rangle}{|\langle x|\GHS|x\rangle|}$ \label{aeq:xyD}
		\STATE $x\gets y$
		\STATE  $\kappa\gets \kappa+1$
		\ENDIF
		\ENDWHILE
		\STATE
\end{algorithmic}
\end{algorithm}

More specifically,
  to aid in the proof of Theorem \ref{thm:main_DT}, we define a new algorithm that couples the continuous and the discrete processes.   The full description is provided in Algorithm \ref{algo:checkCD}.
  Algorithm \ref{algo:checkCD} is derived by modifying specific lines in Algorithm \ref{algo:check} as follows.
    \begin{itemize}
    \item[\textcircled{1}] Change line \ref{aeq:initC} to contain the variables for both the continuous and discrete process:
\begin{center}
\begin{algorithmic}
\STATE $x \gets x_{in}$, $\tau \leftarrow 0$,  $\xi (0)\gets x_{in}$ 
	\STATE $\tau_D\gets 0$, $\xi_D(0)\gets x_{in}$ 
\end{algorithmic}
\end{center}

\item[\textcircled{2}] Delete the sentence in line \ref{aeq:Mt}.
To keep Algorithm \ref{algo:checkCD}  flexible as an analytical tool, we do not assign specific values to $t$ and $M$.
 Accordingly, we modify line \ref{aeq:tauM} to ``While True" which means loop forever. Additionally, we remove the sentences in lines \ref{aeq:cm}, \ref{aeq:cm1} and \ref{aeq:cm2}, and delete line \ref{aeq:end}.

\item[\textcircled{3}] Modify lines \ref{aeq:cr1}, \ref{aeq:cr2} and \ref{aeq:cr3} to
include the updates for both the continuous and discrete processes. In particular, the discrete process is derived by rounding the continuous process:
\begin{center}
\begin{algorithmic}
\STATE Sample $u\in[0,1]$ from  uniform distribution $[0,1]$
\STATE $\Delta \tau\gets \frac{\ln(1/u)}{|\langle x|\GHS|x\rangle|}$
\STATE Set $\xi(s)=x$ for all $s\in (\tau,\tau+\Delta \tau]$
\STATE $\tau\leftarrow \tau+\Delta \tau$
\STATE  Let $\Delta \tau_D$ be the rounded value of $\Delta \tau$, which is 
        the largest value in the set $\{k\delta\}_{k=0,...,K}$ that does not exceed $\Delta \tau$.
\STATE Set $\xi_D(s)=x$ for all $s\in (\tau_D,\tau_D+\Delta\tau_D]$ 
\STATE   $\tau_D \gets \tau_D +\Delta\tau_D$  
\end{algorithmic}
\end{center}
We similarly modify lines \ref{aeq:ol1} and \ref{aeq:ol2}.
    \end{itemize}
Algorithm \ref{algo:checkCD} either returns ``Reject" or loops forever. 
 By Claim \ref{claim:C2D} and the construction of Algorithm \ref{algo:checkCD}, we know that
    \begin{itemize}
         \item  The random variables $(\tau,\xi)$  evolves in the same way as $(\tau,\xi)$ in Algorithm \ref{algo:check}.
        \item  The random variables  $(\tau_D,\xi_D)$ evolves in the same way as $(\tau_D,\xi_D)$ in Algorithm \ref{algo:check}D. 
    \end{itemize}

Here, we outline the key observations needed to prove Theorem \ref{thm:main_DT}. A more formal proof will follow in the next section.   
Roughly speaking, to establish that  Algorithm \ref{algo:check}D is a good discrete approximation of Algorithm \ref{algo:check}, we analyze the differences between the random variables $(\tau_D,\xi_D)$ and $(\tau,\xi)$ in Algorithm \ref{algo:checkCD}. Specifically, we observe that:
 \begin{itemize}
     \item  The sampling of $\Delta\tau$ or  $\Delta\tau_D$ is independent of the other steps and can therefore be analyzed separately.
     \item By construction, $|\Delta\tau- \Delta\tau_D|\leq \delta$  unless $\Delta\tau\geq K\delta$, which is unlikely to happen. When we further set an upper bound $t=poly(n)$ and terminate Algorithm \ref{algo:checkCD} once $\tau>t$, the accumulated error between $\tau_D$ and $\tau$ is most likely within $\delta \times \text{ (number of transitions in $\xi$)}$. 
     \item Note that $\delta$ is exponentially small. As long as the number of transitions is $poly(n)$, we ensure that $|\tau-\tau_D|$ is exponentially small. 
In the Yes instance, Lemma \ref{lem:effS} guarantees a non-trivial probability that the number of transitions is $poly(n)$. In the NO instance, the additional check ``If $\kappa\geq M$ then Return Reject" in Algorithm \ref{algo:check} and Algorithm \ref{algo:check}D ensures that any adversary who successfully cheats can only have $M=poly(n)$ transitions.
 \end{itemize}

\subsubsection{Proof of Theorem \ref{thm:main_DT}}


The formal proof of Theorem \ref{thm:main_DT} is as follows. For  Algorithm \ref{algo:checkCD}, let $\kappa(x_{in},\tau_C)$   be the random variable which denotes the number of transitions in $\xi(s)$ for $s\in [0,\tau_C]$. Similarly, let $\kappa_D(x_{in},\tau_D)$ denote the number of transitions in $\xi_D(s)$ for $s\in [0,\tau_D]$.

To simplify the notation, we use $t_C,t_D$  and $M$ to represent the values specified in Algorithm \ref{algo:check} and Algorithm \ref{algo:check}D: 
\begin{align}
    t_C &:= 8(n+p'(n))/\epsilon,\\
    t_D &:= t_C - 2^{-n},\\
    M &= 2^kmn^3t_C\|H\|.    
\end{align}
Note that the $\delta$ and $K$ from Eqs.~(\ref{eq:Md1})(\ref{eq:Md2}) satisfy
\begin{align*}
    &M\delta \leq 2^{-n},\\
    & M \cdot exp(-\lambda K\delta) \leq 2^{-n},  \text{ for }\lambda\geq 2^{-p'(n)}.
\end{align*}
\begin{proof}[\textbf{of Completeness of Theorem \ref{thm:main_DT}}]
Compared to Algorithm \ref{algo:check}, its discretized version (Algorithm \ref{algo:check}D) only modifies the sampling of the waiting time, replacing  $\Delta\tau$ from a continuous process with its discrete approximation $\Delta\tau_D$. The sampling of $\Delta\tau$ or $\Delta\tau_D$   is independent of the other steps. Thus the Completeness proof of Theorem \ref{thm:main_CT} also works for Theorem \ref{thm:main_DT}, except that we need to bound $$Pr(\kappa_D(x_{in},t_D)\leq M)
\text{
instead of }
Pr(\kappa(x_{in},t_C)\leq M).
$$

According to the randomness in Algorithm \ref{algo:checkCD} we define the following two events.  $Event_1$ is described w.r.t. the variables in the discrete process, and $Event_2$ is described w.r.t. the  variables in the continuous process:
\begin{align*}
    &Event_1:=\{\kappa_D(x_{in},t_D)\leq M\},\\
    &\begin{aligned}
Event_2:= &\left\{ \kappa(x_{in},t_C)\leq M, \text{and for all the}\right.\\
&\text{transitions in $\xi(s)$ for $s\in[0,t_C],$}\\
&\text{none of the $\Delta \tau$ is greater than $K\delta$.}\\
&\text{(Thus $|\Delta\tau-\Delta\tau_D|\leq\delta$)}\}.
\end{aligned}
\end{align*}

Notice that $Event_2$ implies $Event_1$ since 
 $t_D \leq t_C 
 - M \delta.$ Thus we have 
\begin{align}
Pr\left(\kappa_D(x_{in},t_D)
\leq M\right) &\geq   Pr(Event_2).\label{eq:completeness_cp}
\end{align}
To estimate $Pr(Event_2)$, firstly
note that the probability that   a particular $\Delta \tau$ exceeds $K\delta$ is negligible. More specifically, from line \ref{aeq:cr1_CD} in Algorithm \ref{algo:checkCD},  $\Delta\tau$ is sampled from the exponential distribution with parameter $\lambda_x:=|\langle x|\GHS|x\rangle|$. By line \ref{aeq:xxD} we have  $\lambda_x\neq 0$, thus by Remark \ref{remark:G} in Appendix \ref{appendix:precision} we have
$$\lambda_x\geq 2^{-p'(n)}.$$  Thus
w.r.t. a particular $x$, we have 
\begin{align}
    Pr(\Delta\tau\geq K\delta)= \exp(-\lambda_x K\delta) \leq 2^{-n}/M.
\end{align}

Thus begin with Eq.~(\ref{eq:completeness_cp}) and apply a union  bound, we conclude that  
\begin{align}
&Pr\left(\kappa_D(x_{in},t_D)
\leq M\right)\\ &\geq   Pr(Event_2) \nonumber\\
    &\geq Pr(\kappa(x_{in},t_C)\leq M)- M \cdot 2^{-n}/M \nonumber\\
    &\geq 1/2 - 2^{-n}.\nonumber
\end{align}
where $Pr(\kappa(x_{in},t_C)\leq M)\geq 1/2$ is from the completeness proof for Theorem \ref{thm:main_CT}.
\end{proof}

\noindent\begin{proof}[\textbf{of Soundness of Theorem \ref{thm:main_DT}}]
Recall that Algorithm \ref{algo:checkCD} either returns ``Reject" or loops forever. We define $\tau_{rej}$ and $\tau_{rej,D}$ as the random variables which denote the value of $\tau$ and $\tau_D$ respectively, at the moment  Algorithm \ref{algo:checkCD} returns ``Reject". 

Define two events  w.r.t. the variables in the discrete process:
\begin{align*}
     Event_3:=&\{
\tau_{rej,D}\geq t_D, \text{ and }
\kappa_D(x_{in},t_D)\leq M\},\\
Event_4:=&\{
\tau_{rej,D}\geq t_D, \text{ and }
\kappa_D(x_{in},t_D)\leq M\\
& \text{and for all the transitions of $\xi_D(s)$ for}\\
&\text{$s\in [0,t_D]$, none of the $\Delta\tau$ exceeds $K\delta$.}\\
&\text{(Thus $|\Delta\tau-\Delta\tau_D|\leq\delta$)}\}.
\end{align*}
   
By the construction of Algorithm \ref{algo:checkCD} we have 
\begin{align}
    Pr(\text{Algorithm \ref{algo:check}D Accept})=Pr(Event_3).
\end{align}
Since the probability that   a particular $\Delta \tau$ exceeds $K\delta$ is negligible,
similar to  the above Completeness proof of Theorem \ref{thm:main_DT}, using a union bound, we have
\begin{align}
    Pr(Event_3) &\leq Pr(Event_4)+ M\cdot 2^{-n}/M \nonumber\\
    &= Pr(Event_4)+ 2^{-n}.
\end{align}

 Moreover, $Event_4$ implies the event $\{\tau_{rej}\geq t_D - M\delta\}$, which is described w.r.t.  the variables in the continuous process.
Thus in the No instance,
\begin{align}
    &Pr(\text{Algorithm \ref{algo:check}D} \text{ Accept})\nonumber\\
    \leq &Pr(Event_4) +2^{-n}\\
    \leq &Pr(\tau_{rej}\geq t_D-M\delta) + 2^{-n}\nonumber\\
     \leq &Pr(\tau_{rej}\geq t_C-2\cdot 2^{-n})+ 2^{-n}\nonumber\\
    = &Pr(\text{Algorithm \ref{algo:check}}^* \text{ Accept when}\\
    &\text{$t$ is set to } t_C-2\cdot 2^{-n}))+ 2^{-n}.\nonumber
\end{align}
Recall that Algorithm \ref{algo:check}$^*$ is defined in Section \ref{sec:analysis}, which is  Algorithm \ref{algo:check}  without the check ``if $\kappa\geq M$, then Return Reject". Use the same analysis as in Claim \ref{claim:18}, in particular Eqs.~(\ref{eq:star_start})(\ref{eq:star_end})(\ref{eq:star_end_end}), 
 we have that 
\begin{align}
    &Pr(\text{Algorithm \ref{algo:check}}^* \text{ Accept when }\text{$t$ is set to } t_C-2\cdot 2^{-n}))\nonumber\\
    &\leq 2^n\cdot 2^{3p(n)}\cdot (1-\epsilon/2)^{t_C-2\cdot 2^{-n}} \nonumber\\
    &\leq 2\cdot 2^{-n}.\nonumber
\end{align}
    Thus we complete the proof.
\end{proof}

\newpage 
\appendix

\section{Relationship to matrix verification}\label{appendix:matrix_verificaion}
Our techniques for proving Theorem \ref{thm:intro}  work for a slightly general setting, that is, we can generalize from local Hamiltonians to sparse matrices with small norm.   The detailed setting is as follows.

Consider a matrix $H\in\bR^{2^n\times 2^n}$, and a vector $\ket{\phi}\in\bR^{2^n}$. $H$ is $poly(n)$-sparse, that is, every row and every column only have $poly(n)$ non-zero entries. Besides,  $H$ has a  \textbf{small norm},  that is,  $\|H\|\leq poly(n)$. Suppose we have  query access to $H$ and $\ket{\phi}$. That is, 
\begin{itemize}
    \item[(1)] For any $x,y\in\{0,1\}^n$, there is a $poly(n)$ time algorithm which returns $\langle x|H|y\rangle$.
    \item[(2)] For any row index $x$, there is a polynomial time algorithm which outputs the column indices of the non-zero entries in the row,  that is,  $y$ s.t. $\langle x|H|y\rangle\neq 0$. There is a similar algorithm for listing all non-zero entries of a chosen column. 
    \item[(3)]  Given $x$, there is a $poly(n)$-time algorithm which outputs $\langle x|\phi\rangle$ up to a common factor, that is, $c_{\phi}(n)\cdot \langle x|\phi\rangle$ for some unknown $c_{\phi}(n)$. Note that this allows one to efficiently compute the ratio $\frac{\langle x|\phi\rangle}{\langle y|\phi\rangle}$.  
\end{itemize}

The problem is to design an algorithm to distinguish the following two cases, with as few queries as possible:
\begin{itemize}
	\item Yes instance: $\langle \phi|H|\phi\rangle =0$, $\ket{\phi}$ is promised to be the ground state of $H$.
	\item No instance: $\lambda(H)\geq 1/poly(n)$. $\ket{\phi}$ can be arbitrary.
\end{itemize}
Here $\lambda(H)$ is the ground energy of $H$.
\begin{theorem} Under the assumptions in Remark \ref{remark:precision} and Remark \ref{remark:dist} in Appendix \ref{appendix:precision}.
	Given query access to $(H,\ket{\phi})$, where  $(H,\ket{\phi})$ is promised to satisfy either the Yes or No instance. There exists an algorithm $\cA(x)$  which takes an input $x\in\{0,1\}^n$, runs in $poly(n)$ time, and only makes $poly(n)$ queries to $H,\ket{\phi}$ such that
 \begin{itemize}
     \item If  $(H,\ket{\phi})$ is a Yes instance, there exists $x\in\{0,1\}^n$ such that the algorithm accepts with probability $\geq 1/2$.
     \item If it is a No instance, $\forall x,$ the algorithm accepts with exponentially small probability.
 \end{itemize}
\end{theorem}
Roughly speaking $\cA(x)$ is a random walk over $\{0,1\}^n$ which starts from the state $x$.
We omit the proof since it is the same as the proof of Theorem \ref{thm:intro}.

\section{$\MA$-$\hardness$}\label{appendix:MAhard}

In this section, we briefly explain how the proof in Section 4 of \cite{bravyi2006complexity} implies LHP with succinct ground state is $\MA$-$\hard$. We also check the Hamiltonian in the reduction satisfies Remark \ref{remark:precision} in Appendix \ref{appendix:precision}.   Let $p(n)$ be a sufficiently large polynomial.

Consider a problem $L$ in $\MA$, for any instance $x$,
 \cite{bravyi2006complexity}  shows that 
 one can view the $\BPP$ verification circuit  as a quantum circuit. Specifically, 
 
  \begin{definition}[$\MA_{q1}$]
  \label{def:q1}
 	A promise problem $L_{yes},L_{no}\subseteq \Sigma^*$ belongs to $\MA_{q1}$ if there  exists a polynomial $p$ and a poly-size classical reversible circuit $V_x$ that takes  input in $(\bC^2)^{\otimes p(|x|)}$ and is followed by a single qubit measurement, such that
 	\begin{align}
 		& x\in L_{yes} \Rightarrow \exists \ket{\xi}, P\left[ V_x \left(|00...0\rangle,\ket{+}^{\otimes r},\ket{\xi}\right)=1\right] =1,\label{eq:MAh}\\	
 		& x\in L_{no} \Rightarrow \forall \ket{\xi}, P\left[ V_x \left(|00...0\rangle,\ket{+}^{\otimes r},\ket{\xi}\right)=1\right] \leq 1/2.
 	\end{align}
 	Note that w.l.o.g. we can assume that $r$ is even, since adding one $\ket{+}$ state which is independent of other parts of the circuit does not influence the accepting probability.
 \end{definition}
 
 \begin{lemma}\label{lem:MAh}
 	$\MA=\MA_{q1}$. If $x\in L_{yes}$, the $\ket{\xi}$  in Eq.~(\ref{eq:MAh}) can always be chosen to be a computational basis, denoted as $\ket{w}$.
 \end{lemma}
 \begin{proof}
 	Compared to the definition of $\MA_q$ in Definition 7 of \cite{bravyi2006complexity}, $\MA_{q1}$ in Definition \ref{def:q1} requires perfect completeness. It is well known that the complexity class $\MA$ remains unchanged if we require perfect  completeness. Based on this fact, one can check that the  Lemma 2 in \cite{bravyi2006complexity} for proving $\MA=\MA_q$ also works for proving $\MA=\MA_{q1}$.
 \end{proof}

 Begin with the $\MA$-$\Complete$ problem $\MA_{q1}$ with $L_{yes}\cup L_{no}$,  suppose $V_x$ is 
  consisting of $T$ classical reversible gates, denoted as $V_x:=R_T...R_1$. As shown in Section 4 of \cite{bravyi2006complexity},
one can use Kitaev's circuit to Hamiltonian reduction~\cite{kitaev2002classical} to get a Hamiltonian $H_x$, where $x\in L_{yes}$ or $x\in L_{no}$ will be mapped to $H_x$ with $\lambda(H_x)\leq a$ or $\lambda(H_x)\geq b$ where $b-a\geq 1/poly(n)$.  Besides,  one can check that each entry of the Hamiltonian is of form $\frac{N_1}{N_2}$,  and greater than $1/2^{p(n)}$. 

Then we check that the above reduction from $x$ to $H_x$
satisfies the promise of LHP with succinct ground state. That is,  in the Yes case, there always \textit{exists} a succinct ground state.    Since $V_x$ has perfect completeness, one can check that the history state below is a ground state of $H_x$,

\begin{align*}
	\ket{\psi_{hist}}= \frac{1}{T+1}\sum_{t=0}^T \left(R_t...R_1  |00...0\rangle\ket{+}^{\otimes r}\ket{w}\right) \otimes \ket{t},	
\end{align*}

where $\ket{w}$ is the computational basis defined in Lemma~\ref{lem:MAh}, and the clock $\ket{t}$ uses the unary encoding, that is,

\begin{align*}
\ket{t}=|\underbrace{0..0}_\text{$T-t$ zeros,}\underbrace{1..1}
_\text{$t$ ones.}\rangle.
\end{align*}

Given a computational basis $x$, denote the $T$ bits that correspond to the clock register as $x_c$, and the other bits as $x_{o}$. One can check that if $x_c$ is not of the form of unary encoding, then $\langle x|\psi_{hist}\rangle=0$. If $x_c=\ket{t}$, then

\begin{align}
    \langle x| \psi_{hist}\rangle 
    &= \frac{1}{T+1} 
    \langle x_{o}| R_t...R_1  |00...0\rangle\ket{+}^{\otimes r}\ket{w} \nonumber\\
    &= \frac{1}{T+1}  \frac{1}{2^{r/2}} \sum_{y \in\{ 0,1\}^{r}} 
\langle (R_t...R_1)^\dagger x_o|00...0,y,w\rangle.\label{eq:89_appen}
\end{align}

Since $\{R_t\}_t$ are classical reversible gates, one can easily use the poly-size classical circuit to compute $z:=(R_t...R_1)^\dagger x_o$. 
If in $z$ the bits correspond to the ancillas $|00...0\rangle$ and witness $w$ are $|00...0\rangle$ and $\ket{w}$, then Eq.~(\ref{eq:89_appen}) is the sum of one  $\frac{1}{T+1}  \frac{1}{2^{r/2}}$ plus $2^r-1$ zeros, thus

\begin{align}
    \langle x| \psi_{hist}\rangle =\frac{1}{T+1}  \frac{1}{2^{r/2}}.
\end{align}

Otherwise $\langle x| \psi_{hist} \rangle$ is a sum of $2^r$ zeros thus equal to $0$.

Since
$r$ is even by Definition \ref{def:q1}, one can check that   $\langle x|\psi_{hist}\rangle$ 
is of form $\frac{N_1}{N_2}$, and greater than

$$\frac{1}{poly(n)2^{r/2}}\geq 2^{-p(n)}.$$

In summary, 
$\langle x|\psi_{hist}\rangle$   can be computed by  a $poly(n)$-size classical circuit thus is succinct, and satisfies Remark \ref{remark:precision} in Appendix \ref{appendix:precision}. 
Thus LHP with succinct ground state is $\MA$-$\hard$.

\section{Proof of two facts}\label{appendix:fact}
\begin{proof}[of Fact \ref{fact:ex}]
	 For any $x\in[-1,1]$, consider its Taylor series with Lagrange remainder term, we have 
\begin{align}
e^{-x} = 1 - x+ R_1(x),	
\end{align}
 where $R_1(x)=\frac{e^{-\eta}}{2!}x^2$ for  $\eta\in [-1,1]$. Thus 
 \begin{align}
 	|e^{-x} - (1 -x)| \leq  \frac{e}{2}x^2\leq 2x^2.
 \end{align}
\end{proof}

\noindent\begin{proof}[of Fact \ref{fact:prin}]
 W.l.o.g.,  we  write $M$ as
   \begin{align}
   M = \begin{bmatrix}
   	  N & E\\
   	  E^\dagger &F
   \end{bmatrix},	
   \end{align}
where $E\in \bR^{|S'|\times (|S|-|S'|)}$, $F\in \bR^{(|S|-|S'|)\times (|S|-|S'|)}$. 

Since $M,N$ are Hermitians,  they can be diagonalized by orthogonal basis. In particular, for $M$, for any normalized state $\ket{\psi}$, we have 
	\begin{align}
		\lambda(M)\leq \langle \psi|M|\psi\rangle \leq \lambda_{max}(M). \label{eq:max}
	\end{align}

 Let $\ket{\eta}\in \bR^{|S'|}$ be a normalized  eigenvector of $N$ with eigenvalue $\alpha$. Let $\ket{\eta0}\in \bR^{|S|}$ to be the state that extending $\ket{\eta}$ to $\bR^{|S|}$ by adding $0$ values to entries in $S\backslash S'$.  $\ket{\eta0}$ is normalized by definition.
One can verify that 
\begin{align}
	\langle \eta 0|	M| \eta0\rangle &= \langle \eta |	N |\eta\rangle\\
	&=\alpha.
\end{align}
Together with inequality (\ref{eq:max}), we have 
\begin{align}
	\lambda(M)\leq \alpha \leq \lambda_{max}(M).
\end{align}
Thus
\begin{align}
	\lambda(M)\leq \lambda(N)\leq \lambda_{max}(N)\leq \lambda_{max}(M).	
\end{align}

\end{proof}

\section{Properties of $\FH$.}\label{appendix:FN}

\begin{proof}[of Lemma \ref{lem:basic_FH} from \cite{ten1995proof} \cite{bravyi2022rapidly}]
 Note that $H$ is real-valued and Hermitian, thus $H$ is symmetric.
	(1) is true by definition of stoquastic.
	For (2), one can verify that for any $x$,
  \begin{align*}
  	  \langle x|\FH|\phi\rangle & = \sum_y \langle x|\FH|y\rangle \langle y|\phi\rangle \\
  	  & =\langle x|\FH|x\rangle \langle x|\phi\rangle + \sum_{y:(x,y)\in S^{-}} \langle x|H|y\rangle \langle y|\phi\rangle\\
  	  &= \langle x|H|x\rangle \langle x|\phi\rangle + \sum_{z:(x,z)\in S^+} \langle x|H|z\rangle \langle z|\phi\rangle\\
     &+\sum_{y:(x,y)\in S^{-}} \langle x|H|y\rangle \langle y|\phi\rangle\\
  	  &= \sum_{y} \langle x|H|y\rangle \langle y|\phi\rangle\\
  	  &= \langle x|H|\phi\rangle.
  \end{align*}
  Thus $\FH\ket{\phi}=H\ket{\phi}$.	
	 For (3), 	by (1) $\FH$ is Hermitian and thus diagonalizable, thus $\lambda(\FH)$ is well-defined.
		Consider any complex-valued state $\ket{\xi}$, one have
	\begin{align}
		&\langle \xi|\FH-H|\xi\rangle \nonumber\\ &= \sum_{x,y} \langle \xi|x\rangle \langle x|\FH-H|y\rangle \langle y|\xi\rangle\nonumber\\
		&= \sum_{(x,y)\in S^+} \langle \xi|x\rangle (-\langle x|H|y\rangle) \langle y|\xi\rangle
		+ \sum_{x} \langle \xi|x\rangle \langle x|\FH-H|x\rangle\langle x|\xi\rangle  \nonumber\\
		&=\! \sum_{(x,y)\in S^+} \langle \xi|x\rangle (-\langle x|H|y\rangle) \langle y|\xi\rangle\nonumber+ \sum_{x} \langle \xi|x\rangle \!\!\!\!\!\!\! \sum_{ y:(x,y)\in S^+}  \!\!\!\!\!  \langle x|H|y\rangle \frac{\langle y|\phi\rangle}{\langle x|\phi\rangle} \langle x|\xi\rangle\nonumber\\
		&=\!\sum_{(x,y)\in S^+} \langle x|H|y\rangle \left[    \frac{\langle y|\phi\rangle \langle \xi|x\rangle \langle x|\xi\rangle}{\langle x|\phi\rangle} - \langle\xi|x\rangle \langle y|\xi\rangle\right].\label{eq:81}
	\end{align}
	For any $x,y$, define $$s(x,y):=sign(\langle x|H|y\rangle).$$ Note that $H$ is symmetric\footnote{Since $H$ is real valued and Hermitian.}, thus $$s(x,y)=s(y,x), \langle x|H|y\rangle =\langle y|H|x\rangle.$$
Using the definition of $S^+$, one gets Eq.~(\ref{eq:81}) equals to
\begin{align}
&=\!\!\!\!\!	\sum_{(x,y)\in S^+} |\langle x|H|y\rangle | \left[  \frac{|\langle y|\phi\rangle|}{|\langle x|\phi\rangle|} \langle \xi|x\rangle \langle x|\xi\rangle - s(x,y)\langle\xi|x\rangle \langle y|\xi\rangle\right]\nonumber\\
&=\!\!\!\!\! \sum_{(x,y)\in S^+} \frac{1}{2} |\langle x|H|y\rangle | 
\left| \sqrt{ \frac{|\langle y|\phi\rangle|}{|\langle x|\phi\rangle|}} \langle \xi|x\rangle - s(x,y)\sqrt{ \frac{|\langle x|\phi\rangle|}{|\langle y|\phi\rangle|}} \langle \xi|y\rangle\right|^2\nonumber\\
&\geq 0. \label{eq:eq}
\end{align}
 In other words, $H$ can always achieve smaller energy than $\FH$, thus
 \begin{align}
 	\lambda(\FH)\geq \lambda(H).\label{eq:ee}
 \end{align}
  Notice that when $\ket{\xi}=\ket{\phi}$, Eq.~(\ref{eq:eq}) equals to 0. Thus if further $\ket{\phi}$ is the ground state of $H$, then $\FH$ can achieve the energy $\lambda(H)$  w.r.t. $\ket{\phi}$, then $\ket{\phi}$ is the ground state of $\FH$ 
  due to Eq.~(\ref{eq:ee}).
\end{proof}

\section{Properties of the CTMC}\label{appendix:GH}

\begin{proof}[of Corollary \ref{cor:basic_GH}.]
The proof comes from \cite{bravyi2022rapidly}.
For (1), it suffices to notice that 
\begin{align*}
\GHt = Diag(\phi)^{-1}  (\lambda(\FH)I-\FH)Diag(\phi),
\end{align*}
	Note that $Diag(\phi)^{-1}$ is well-defined since $\phi$ is regularized and real-valued.

For (2), one can verify $\langle y| \GHt|x\rangle\geq 0$ for $y\neq x$ by definition. If $\ket{\phi}$ is a ground state of $H$,

\begin{align}
\sum_y 	\langle y| \GHt|x\rangle & = \lambda(\FH)-  \sum_{y} 	\langle y| \FH|x\rangle \frac{\langle y|\phi\rangle}{\langle x|\phi\rangle}\nonumber\\
&=\lambda(\FH)\! - \!\langle \phi |\FH|x\rangle \frac{1}{\langle x|\phi\rangle}\nonumber\\
&=\lambda(\FH)\! -\! \lambda(\FH)\langle \phi |x\rangle \frac{1}{\langle x|\phi\rangle}\label{eq:26}\\
&=0.\nonumber
\end{align}

 where Eq.~(\ref{eq:26}) is from  Lemma \ref{lem:basic_FH} (3) and the fact that $F$ is Hermitian, $\lambda(\FH)$ must be real-valued as eigenvalues of Hermitian $\FH$.
 
For (3), when $\ket{\phi}$ is ground state of $H$, by (2) we know $\GHt$ is a legal generator. One can verify that for any $y$,
\begin{align*}
	&\sum_x 	\langle y| \GHt|x\rangle \pi(x)\\
 & = \lambda(\FH)\pi(y)-  \sum_{x} 	\langle y| \FH|x\rangle \frac{\langle y|\phi\rangle\langle x|\phi\rangle}{c}
	\\
&=\lambda(\FH)\pi(y) - \langle y |\FH|\phi\rangle \frac{\langle y|\phi\rangle}{c}\\
&=\lambda(\FH)\pi(y) - \lambda(\FH)\langle y |\phi\rangle \frac{\langle y|\phi\rangle}{c}\\
&=0.
\end{align*}

\end{proof}

\begin{proof}[of Lemma \ref{lem:effS}]
This proof is modified from \cite{bravyi2022rapidly}. First note that since $\ket{\phi}$ is the ground state of $H$, by Corollary \ref{cor:basic_GH} (2), $\GHt$ is a legal generator of a CTMC, thus Algorithm \ref{algo:Gill} w.r.t. $\GHt$ is well-defined.

	Define $c=\|\ket{\phi}\|^2$ and $\pi(x)=|\langle x|\phi\rangle|^2/c$.   Let $h$ be an infinitesimal value.
 
 For any random process generating a random variable  $\xi:[0,t]\rightarrow S,$ define $T(\tau,x,h)[\xi]$ as the number of transitions  occurring within the time interval $[\tau,\tau+h]$, conditioned on $\xi(\tau)=x$.  Let $I_{bad}[\xi]$ be the indicator function that at least $2$ transitions  happen in  any of the time intervals 
 $$[0,h],[h,2h],....,[t-h,t]\text{ for } t=poly(n).$$

 Here for simplicity we assume that $t/h$ is an integer.

	Let  $\rho(\xi)$ be the distribution of $\xi:[0,t]\rightarrow S$ generated by running Algorithm \ref{algo:Gill} w.r.t. $(\GHt,x_{in},t)$, where $x_{in}$ is sampled
	 from distribution $\pi$. 	 
	  By Lemma \ref{lem:legal} we know
\begin{align*}
    Pr_{\xi\sim\rho(\xi)}(T(\tau,x,h)[\xi]\geq 2)=O(h^2), \forall x,\tau.
\end{align*}
	 By a union bound we know 
	 \begin{align}
	 &Pr_{\xi\sim\rho(\xi)}(I_{bad}[\xi]=1)\nonumber\\
  &\leq \sum_{j=0}^{t/h} \sum_x Pr_{\xi\sim \rho(\xi)}(\xi(jh)=x)Pr_{\xi\sim\rho(\xi)}(T(jh,x,h)[\xi]\geq 2)	\nonumber\\
	 &= \frac{t}{h}  O(h^2)\nonumber\\
	 &=O(h).\label{eq:36}
	 \end{align}
	  By Corollary \ref{cor:basic_GH} (3), we know that $\pi$ is a stationary distribution of $\GHt$. Thus for $\xi\sim\rho(\xi)$, at any time $s\in[0,t]$, the distribution of $\xi(s)$ is given by $\pi$, which we denote as $\pi_s=\pi$.
	 For $s\leq t,$
	 let $\pi_{good,s},\pi_{bad,s}$ be the distribution of $\xi(s)$ conditioned on 
	``good", ``bad", that is,
	\begin{align}
		\pi_{good,s}(x) &= Pr_{\xi\sim\rho(\xi)}(\xi(s)=x|I_{bad}(\xi)=0),\label{eq:pi}\\
		\pi_{bad,s}(x) &= Pr_{\xi\sim\rho(\xi)}(\xi(s)=x|I_{bad}(\xi)=1).\nonumber
	\end{align}
	 we have
	 \begin{align*}
	 \pi_s(x) = &Pr_{\xi\sim\rho(\xi)}\left(I_{bad}[\xi]=0\right) \pi_{good,s}(x) \\&+ Pr_{\xi\sim\rho(\xi)}\left(I_{bad}[\xi]=1\right) \pi_{bad,s}(x).
	 \end{align*}
	 Thus
	 \begin{align}
	 \pi_{good,s}(x)-\pi(x)=O(h).\label{eq:116}
	 \end{align}
	 In other words, $\pi_{good,s}$ is almost $\pi$.

\vspace{1em}

For $s\leq t$, where $s$ is an integer multiple of $h$, let $T_{good}(s)[\xi]$ be the number of transitions  in time [0,s], conditioned on $I_{bad}[\xi]=0$. We have that\footnote{
In Eq.~(\ref{eq:412}), we use 
\begin{align}
&Pr_{\xi\sim\rho(\xi)}\!\left(T(s,x,h)[\xi]=1,\xi(s+h)=y|\xi(s)=x,I_{bad}[\xi]=0\right)\nonumber\\
=&Pr_{\xi\sim\rho(\xi)}\!\left(T(s,x,h)[\xi]=1,\xi(s+h)=y|\xi(s)=x,T(s,x,h)[\xi]\leq 1\right) \label{eq:E7}\\
=&Pr_{\xi\sim\rho(\xi)}\!\left(T(s,x,h)[\xi]=1,\xi(s+h)=y|\xi(s)=x\right)\nonumber\\
&\div Pr_{\xi\sim\rho(\xi)}\left(T(s,x,h)[\xi]\leq 1|\xi(s)=x\right)\label{eq:E8}\\
=& \left(\langle y|\GHt|x\rangle h + O(h^2)\right)\div \left(1-O(h^2)\right)\nonumber\\
=& \langle y|\GHt|x\rangle h + O(h^2).\nonumber
\end{align}
Here Eq.~(\ref{eq:E7}) comes from the Markov property, that is, the events occur within the time interval $[s,s+h]$ are independent from the events occur outside this time interval. Eq.~(\ref{eq:E8}) comes from Bayesian rule.

} 
 \begin{align}
 	&E_{\xi\sim\rho(\xi)}[T_{good}(s+h)[\xi]-T_{good}(s)[\xi]]\nonumber\\ &=\sum_{x} \pi_{good,s}(x) \cdot 1\cdot \left(\sum_{y\neq x} \langle y|\GHt|x\rangle h + O(h^2)\right)\label{eq:412}\\
 	&= \sum_x \pi(x)\sum_{y:y\neq x} \langle y|\GHt|x\rangle h + O(h^2)\nonumber\\
 	&= \sum_x \frac{|\langle x|\phi\rangle|^2}{c}
	\sum_{y:y\neq x} -\langle y|\FH|x\rangle\frac{\langle  y|\phi\rangle}{\langle x|\phi\rangle}h + O(h^2)\nonumber\\
	&=  \sum_x 
	\sum_{y:y\neq x} -\langle y|\FH|x\rangle\frac{\langle  y|\phi\rangle\langle x|\phi\rangle}{c}h + O(h^2)\nonumber\\
	&= \sum_{(x,y)\in S^{-}} -\langle x|H|y\rangle   
	\langle  y|\phi\rangle\langle x|\phi\rangle \frac{1}{c} h +O(h^2).\label{eq:113}
 	 \end{align}

 Thus 
 \begin{align*}
 	&E_{\xi\sim\rho(\xi)}(T_{good}(t)[\xi])\\
  &= \sum_{(x,y)\in S^{-}} -\langle x|H|y\rangle   
	\langle  y|\phi\rangle\langle x|\phi\rangle \frac{1}{c}  t + \frac{t}{h}O(h^2). \\
 \end{align*}
 Then
  \begin{align}
	 &E_{\xi\sim\rho(\xi)}[T_{good}(t)[\xi]]\nonumber\\
  &\leq  \sum_{(x,y)\in S^{-}} |\langle x|H|y\rangle|\cdot   
	|\langle  y|\phi\rangle|\cdot |\langle x|\phi\rangle| \frac{1}{c}  t + O(h)\nonumber\\
	&\leq  \sum_{(x,y)} |\langle x|H|y\rangle|\cdot   
	|\langle  y|\phi\rangle|\cdot |\langle x|\phi\rangle| \frac{1}{c}  t + O(h)\nonumber \\
	&\leq  \|H\|\frac{1}{c}  t
\left(\sum_{(x,y): \langle x|H|y\rangle \neq 0}
	|\langle  y|\phi\rangle|\cdot |\langle x|\phi\rangle| 	\right)  + O(h)\nonumber \\
	&\leq  \|H\|\frac{1}{c}  t
\left(\sum_{(x,y): \langle x|H|y\rangle \neq 0}
	\!\!\!\!\!\!\!\!\!\!\! |\langle  y|\phi\rangle|^2 \right)^{\frac{1}{2}}\!\!\! \left(\sum_{(x,y): \langle x|H|y\rangle \neq 0}
	\!\!\!\!\!\!\!\!\!\!\! |\langle  x|\phi\rangle|^2 \right)^{\frac{1}{2}} \nonumber\\
 &+ O(h)\nonumber \\
	&  \leq\|H\|\frac{1}{c}  t d\|\phi\|^2 + O(h)\label{eq:499}
	\\
	& =dt\|H\| + O(h),\nonumber
 \end{align}
	where Eq.~(\ref{eq:499}) comes from the fact that since $H$ is $d$-sparse,  for every $y$, we know that $|\langle  y|\phi\rangle|^2$ appears for at most $d$ times  in $\left(\sum_{(x,y): \langle x|H|y\rangle \neq 0}
	|\langle  y|\phi\rangle|^2 \right)$.
	 
Let $\kappa_{good}(x,t)[\xi]$ be  the number of transitions  between $[0,t]$ conditioned on $$\xi(0)=x,I_{bad}[\xi]=0.$$ Let $A(x)$ be the distribution of $\xi$ generated by Algorithm \ref{algo:Gill}  w.r.t. $(\GHt,x,t)$.
 note that
\begin{align}\label{eq:E}
	 &E_{\xi\sim\rho(\xi)}(T_{good}(t)[\xi])\nonumber\\
     &=  \sum_x \pi(x)  
	 E_{\xi\sim A(x)}
	 (\kappa_{good}(x,t)[\xi]). 
\end{align}
Combine Eqs.~(\ref{eq:499})(\ref{eq:E}), using an average argument, we know that there exists an $\hat{x}_{in}$  with
$\pi(\hat{x}_{in})\neq 0$, that is, $ \langle \hat{x}_{in}|\phi\rangle\neq 0$,
such that 
\begin{align}
	E_{\xi\sim A(\hat{x}_{in})} (\kappa_{good}(\hat{x}_{in},t)[\xi])\leq dt\|H\| + O(h).
\end{align}
By a Markov bound, we know 
\begin{align}
	&Pr_{\xi\sim A(\hat{x}_{in})}  \left( \kappa_{good}(\hat{x}_{in},t)[\xi]\geq dn^3t\|H\|\right)\nonumber\\
    &\leq 1/n^2.
\end{align}
Besides, we have that
\begin{align*}
	&Pr_{\xi\sim A(\hat{x}_{in})}  (\kappa(\hat{x}_{in},t)[\xi] \geq dn^3t\|H\|)\\ &= Pr_{\xi\sim A(\hat{x}_{in})}(I_{bad}[\xi]=1,\kappa(\hat{x}_{in},t)[\xi] \geq dn^3t\|H\|)\\ 
	&+ Pr_{\xi\sim A(\hat{x}_{in})}(I_{bad}[\xi]=0,\kappa(\hat{x}_{in},t)[\xi] \geq dn^3t\|H\|).\end{align*}
    
Note that $Pr_{\xi\sim A(\hat{x}_{in})}(I_{bad}[\xi]=1)=O(h) $ by a similar argument as Eq.~(\ref{eq:36}), we finally conclude that

\begin{align}
	&Pr_{\xi\sim A(\hat{x}_{in})}(\kappa(\hat{x}_{in},t)[\xi] \geq dn^3t\|H\|)\nonumber\\
	\leq &Pr_{\xi\sim A(\hat{x}_{in})}(I_{bad}[\xi]=0,\kappa(\hat{x}_{in},t)[\xi]\geq dn^3t\|H\|)+O(h)\nonumber\\
	 =&Pr_{\xi\sim A(\hat{x}_{in})}(I_{bad}[\xi]=0)\times \nonumber\\
 &
	Pr_{\xi\sim A(\hat{x}_{in})}\!  \left(\kappa_{good}(\hat{x}_{in},t\!)[\xi]\!\geq\! dn^3t\|H\| \,|\, I_{bad}[\xi]=0\right) +O(h)\nonumber\\
	\leq& (1-O(h)) 1/n^2 +O(h)\nonumber\\
	\leq& 1/2.
\end{align}

In other word,
\begin{align}
	Pr_{\xi\sim A(\hat{x}_{in})}(\kappa(\hat{x}_{in},t)[\xi] \leq dn^3t\|H\|)
	\geq 1/2.
\end{align}

\end{proof}

\section{Proof of Claim \ref{cla:pro}} \label{appendix:df}

\begin{proof}[of  Claim \ref{cla:pro}]
	First we define more notations for the proof. Conditioned on $\tau_{end}>s,\xi(s)=x$, let $t_x$ be the time that $\xi$ last arrives $x$ before time $s$,  that is,  
	\begin{align}
		t_x:=\max \{\eta\leq s: \xi(\eta)\neq x\}.
	\end{align}
 Let $\Tx$ be the waiting time sampled by line \ref{line:dt} when $\xi$ reaches $x$ at time $t_x$. Note that $c_1 = t_x +\Tx$. 
	If $y\in S_{good}$ is the state that $\xi$ visit next,  that is,  $\xi(c_1)$, similarly define $\Ty$ be the waiting time  by line \ref{line:dt} when $\xi$ reaches $y$ at time $c_1$. 
 
 To simplify the notations, we denote 
 \begin{align*}
     & u_x:=|\langle x|\GHS|x\rangle|,\\
     &P_{xy}:=\langle y|\GHS|x\rangle/u_x.
 \end{align*}

To ease notations, we abbreviate the probability density function
  $p(t_x=s_x |\tau_{end}>s,\xi(s)=x)\, d s_x$ as $p(t_x=s_x) d s_x$.
Note that by definition
\begin{align*}
	\int_{s_x\leq s} p(t_x=s_x)\, d s_x=1.
\end{align*}

To complete the proof, notice that
\begin{itemize}
\item
For Eq.~(\ref{eq:xx}) we can perform calculation in  Appendix \ref{appendix:eq}  \textcircled{A}.

\item
For Eq.~(\ref{eq:xyb}), when $y\not\in S_{good}$, 	 we can perform calculation in  Appendix \ref{appendix:eq} \textcircled{B}.

\item
For Eq.~(\ref{eq:xyg}), when $y\in S_{good}$,	 we can perform calculation in  Appendix \ref{appendix:eq} \textcircled{C}.

\item
Eq.~(\ref{eq:el}) is one minus of the probability of Eqs.~(\ref{eq:xx})(\ref{eq:xyb})(\ref{eq:xyg}). Since $\sum_y\langle y|\GHS|x\rangle =0$, thus we know  Eq.~(\ref{eq:el}) holds.
\end{itemize}
\end{proof}

\vspace{3em}
\section{Remarks on precision}\label{appendix:precision}

\begin{remark}[How we represent values]\label{remark:precision}   \end{remark}

 We say  $x\in\bC$ is represented by $p(n)$-bits, if $x$ is of the form $\frac{N_1}{N_2}+\frac{N_3}{N_4}i$, where $\forall i$, $N_i$ is an integer and $|N_i|\leq 2^{p(n)}$, thus can be represented by $p(n)$ binary bits.  
 
 Note that LHP with succinct ground state is a promise problem, we implicitly assume that there is a sufficiently large polynomial $p(n)=poly(n)$, such that every value in Definition~\ref{def:Hss},  that is,  $\langle x|H|y\rangle,a,b,m, C_{\psi}(x)$ can be represented  by $p(n)$-bits. 	Note that those assumptions implicitly imply
 	\begin{itemize}
 		\item[(1)]  $C_{\psi}(x)$ can be computed \textbf{exactly}. Thus the ratio of the amplitudes,  that is,  $\frac{\langle x|\psi\rangle}{\langle y|\psi\rangle}$, can be computed exactly for $y$ where $C_{\psi}(y)\neq 0$.
 		\item[(2)] If $C_{\psi}(x) \neq 0$, then $|C_{\psi}(x)| \geq 1/2^{p(n)}$.  
 		Similarly for $\langle x|H|y\rangle$.
 		\item[(3)] $\lambda(H)$ can be represented exactly by $poly(n)$ bits, since $\langle x|H|y\rangle, C_{\psi}(x)$  can be represented by $p(n)$ bits, and $\lambda(H)= \frac{\sum_y \langle x|H|y\rangle C_{\psi}(y)}{C_{\psi}(x)}$ for some $x$ s.t. $C_{\psi}(x)\neq 0$.
 	\end{itemize}
For $S\subseteq \{0,1\}^n$,
  we say a matrix $G\in \bC^{|S|\times |S|}$ can be represented by $p(n)$-bits, if all entries $\langle x|G|y\rangle$ can be represented by $p(n)$-bits.

  One can check that if $G$ can be represented by $p(n)/2$-bits, then there exists $p'(n)=poly(n)$ such that\footnote{Given any matrix $M\in\bR^{|S|\times |S|}$, for  any normalized vector $\ket{\phi}\in\bR^{|S|}$, 
since $|\langle x|\phi\rangle|\leq 1, \forall x$, 
we have $ \langle \phi|M^\dagger M|\phi\rangle  = \sum_{x,y,z} \langle \phi|x\rangle \langle x|M^\dagger| y\rangle \langle y|M|z\rangle \langle z|\phi\rangle \leq amax(M)^2 2^{3n}.$ thus $\|M\|\leq  \sqrt{amax(M)^2 2^{3n}}\leq amax(M)2^{2n}.$}
$$amax(G)\leq 2^{p'(n)}, \, \|G\|\leq 2^{p'(n)+2n}.$$
Besides, if $\langle x|G|y\rangle \neq 0$, then $\langle x|G|y\rangle \geq 1/2^{p'(n)}$.

\begin{remark}[Assumptions for sampling]\label{remark:dist} \end{remark}

Let $M=poly(n)$.
Set parameters
\begin{align*}
    \delta &:= 2^{-2n}/M\\
     K &:=  \lceil 2^{2n}  M  \left(n \ln 2 + \ln M\right)) \rceil
\end{align*}

Note that for any sufficiently large  $q=poly(n)$, there is a $poly(n)$-time algorithm whose output distribution 
approximates the truncated discretized exponential distribution $\cD_{K,\delta,\lambda}$ within total variation distance\footnote{Given two probability distribution $p,q$ over a discrete set $\Omega$, the total variation distance $d_{TV}(p,q)$ between $p$ and $q$ is defined as $d_{TV}(p,q)=\sum_{x\in\Omega} \frac{1}{2}|p(x)-q(x)|$, where $p(x)$ is the probability of $p=x$ and similarly for $q(x)$. } $2^{-q}$. 
Specifically, the algorithm is sampling $q^2$ random bits 
$s_1,...,s_q\in\{0,1\}.$ Let

$$\eta= \sum_{j=1}^{q^2} s_j 2^{j-1}.$$

\begin{itemize}
    \item If 
$\eta\in ( \exp(-\lambda (k+1)\delta),\exp(-\lambda k\delta)]$ for $k\leq K$, output $w=k\delta$.
\item  If $\eta \leq \exp(-\lambda K\delta)$, output $w=K\delta$.
\end{itemize}
 The total variation distance between the output distribution and $\cD_{K,\delta,\lambda}$ is $O(\frac{1}{2^{q^2}}K) = O(2^{-q})$ for sufficiently large $q=poly(n).$
  
To ease analysis, in this manuscript we assume that for the parameters defined above, we can use $poly(n)$-time to sample the discretized exponential distribution $\cD_{K,\delta,\lambda}$ \textbf{exactly}.

\begin{remark}[More on how we represent values]\label{remark:G}\end{remark} 

As in Remark \ref{remark:precision}, in the following sections, we assume values in $H,\ket{\phi}$ can be represented by $p(n)$-bits.
Note that by Remark \ref{remark:precision},  if $\langle x_{in}|\phi\rangle\neq 0$, then
\begin{align}
 \left|\frac{ \langle z|\phi \rangle }{\langle x_{in}|\phi\rangle}\right|	=  \left| \frac{ C_{\phi}(z) }{C_{\phi}(x_{in})} \right|
 \leq \frac{2^{p(n)}+2^{p(n)}}{1/2^{p(n)}}
 \leq 2^{3p(n)}.
\end{align}

We represent $\langle x|\FHS|y\rangle$,$\langle x|\GHS|y\rangle$ in a similar way as 
	in Remark \ref{remark:precision}.  Note that  $H_S$ is $2^k m$-sparse where $k$ is a constant and $m=poly(n)$, thus there exists another polynomial $p'(n)$, such that $\langle x|\FHS|y\rangle$, $\langle x|\GHS|y\rangle$ can be represented by $p'(n)/2$ bits, where $$p'(n)/2=O(\log[ (2^{p(n)*2^km}])=O(poly(n)).$$  In particular,   we have 
	\begin{itemize}
		\item[(1)] The followings holds:
  \begin{align*}
      &amax(\GHS)\leq 2^{p'(n)},\\
      &\|\GHS\|,\|\FHS\|\leq 2^{p'(n)+2n}.
  \end{align*}
		\item[(2)] If 
	$\langle x|\GHS|y\rangle\neq 0$, then $\langle x|\GHS|y\rangle\geq 1/2^{p'(n)}$.
	\end{itemize}
	We always assume $p(n),p'(n)$ are sufficiently large polynomials, and $p(n)\ll p'(n)$.

\section{Calculation for Equations}\label{appendix:eq}
Please see the next page.

\begin{figure*}
\begin{align}
\textcircled{A}\quad& \text{Calculation for Claim \ref{cla:pro} Eq.~(\ref{eq:xx}).} \nonumber\\
&\nonumber\\
 &Pr(c_1\geq s+h | \tau_{end}>s,\xi(s)=x)\\ 
	 & = \int_{s_x\leq s}  p(t_x=s_x)  Pr(\Tx \geq s+h-s_x|\Tx\geq s-s_x) d s_x\\
	& = \int_{s_x\leq s}  p(t_x=s_x)  \exp(-u_xh)d s_x\\
	& = 1- u_x h+O(h^2)\\
	&= 1-|\langle x|\GHS|x\rangle|h + O(h^2).\\
 &\nonumber\\
\textcircled{B}\quad& \text{Calculation for Claim \ref{cla:pro} Eq.~(\ref{eq:xyb}).}\nonumber\\
 &\nonumber\\
 & Pr(c_1\leq s+h, \xi(c_1)=y|\tau_{end}>s,\xi(s)=x)\\
 & =
 Pr(\Tx\leq s+h-s_x, \xi(s_x+\Tx)=y|\tau_{end}>s,\xi(s)=x)\\
 &=\int_{s_x\leq s}  p(t_x=s_x) \int_{s-s_x\leq f\leq s+h-s_x}  p(\Tx=f |\Tx\geq s-s_x) P_{xy} df \, d s_x\\
 & =   P_{xy} \int_{s_x\leq s} p(t_x=s_x)  \int_{s-s_x\leq f\leq s+h-s_x}  u_x\exp\left[-u_x (f-s+s_x) \right]
  df \, d s_x\\
  & =  P_{xy} \int_{s_x\leq s}  p(t_x=s_x)  \left[1-\exp(-u_x h)\right]
  d s_x\\
  & =p_{xy} u_x h + O(h^2)\\
  & = \langle y|\GHS|x\rangle h +O(h^2).\\
  &\nonumber\\
\textcircled{C}\quad & \text{Calculation for Claim \ref{cla:pro} Eq.~(\ref{eq:xyg}).}\nonumber\\
  &\nonumber\\
  &Pr(c_1\leq s+h, \xi(c_1)=y, c_2\geq s+h|\tau_{end}>s,\xi(s)=x)\\
 & =
 Pr(\Tx\leq s+h-s_x, \xi(s_x+\Tx)=y, s_x+\Tx+\Ty\geq s+h|\tau_{end}>s,\xi(s)=x)\\
 & =\int_{s_x\leq s}  p(t_x=s_x) \int_{s-s_x\leq f\leq s+h-s_x} p(\Tx=f |\Tx\geq s-s_x) P_{xy} Pr(\Ty\geq h- (f-s+s_x)) df \, d s_x\\
 & =   P_{xy} \int_{s_x\leq s} p(t_x=s_x)  \int_{s-s_x\leq f\leq s+h-s_x}  u_x\exp\left[-u_x (f-s+s_x) \right] \exp\left[ -u_y (h-(f-s+s_x)) \right]
  df \, d s_x\\
  & =   P_{xy} \int_{s_x\leq s} p(t_x=s_x)  \int_{s-s_x\leq f\leq s+h-s_x}  u_x\exp\left[-u_yh\right] \exp\left[ (u_y-u_x) (f-s+s_x) \right]
  df \, d s_x\\
  & = P_{xy} \int_{s_x\leq s} p(t_x=s_x)   u_x\exp\left[-u_yh\right] \frac{1}{u_y-u_x}\left[ \exp\left[(u_y-u_x)h\right]-1 \right] d s_x\\
  & = P_{xy} u_x \exp[-u_yh]h\\
  & = \langle y|\GHS|x\rangle h+ O(h^2).\\
  &\nonumber
	\end{align}
\end{figure*}

\clearpage

\section{Acknowledgement}
We thank Jiayu Zhang,  Thomas Vidick, Urmila Mahadev, Sevag Gharibian,   Yupan Liu and Jielun Chen (Chris)  for their helpful discussions.   We thank the anonymous reviewers for their valuable
suggestions for presentations and open questions.
Jiaqing Jiang is supported by MURI Grant FA9550-18-1-0161 and the IQIM, an
NSF Physics Frontiers Center (NSF Grant PHY-1125565).

\bibliographystyle{ieeetr}
\bibliography{ref.bib}

\end{document}